\documentclass[cleveref,thm-restate]{lipics-v2021}
\hideLIPIcs
\nolinenumbers
\usepackage[all,defaultlines=3]{nowidow}

\usepackage[T1]{fontenc}
\usepackage[utf8]{inputenc}

\newcommand\shortlongversion[2]{#1}
\renewcommand\shortlongversion[2]{#2}
 
\usepackage{xcolor}
\usepackage{amssymb}
\usepackage{mathtools}
\usepackage{amsfonts}
\usepackage{amsmath}
\usepackage{amsthm}
\usepackage{complexity}
\usepackage{mathrsfs}
\usepackage{graphics}
\usepackage{microtype}
\usepackage[many]{tcolorbox}
\tcbuselibrary{theorems}
\usepackage{xspace}
\usepackage{makecell}
\usepackage{pifont}
\usepackage{bm}
\usepackage{scalerel}
\usepackage{tabularx}
\usepackage{afterpage}
\newtheorem{obs}[theorem]{Observation}
\crefname{obs}{observation}{observations}
\Crefname{obs}{Observation}{Observations}

\usepackage{subcaption}
\usepackage{float}
\usepackage{tikz}
\usetikzlibrary{
  calc,arrows,automata,fit,shapes,arrows.meta,positioning,
  decorations.pathmorphing,patterns,shapes.geometric,
  shapes.symbols,shapes.arrows,shapes.misc,
  decorations.shapes,decorations.pathreplacing
}
\usepackage{todonotes}

\definecolor{darkgreen}{rgb}{0.01,0.53,0.47}
\definecolor{darkred}{rgb}{0.62,0.13,0.28}
\definecolor{lightred}{RGB}{222,155,167}

\definecolor{myPurple}{RGB}{230,230,255}
\definecolor{interesting}{RGB}{140,0,60}
\definecolor{brown}{RGB}{210,180,140}

\newcommand{\out}{\operatorname{out}}
\newcommand{\inw}{\operatorname{in}}

\tikzset{
  paramRed/.style={
    draw, rectangle, fill=gray!20, font=\small, rounded corners=3pt,
    minimum width=2cm, align=center},
  paramGreen/.style={
    draw, rectangle, fill={rgb,255:red,255; green,230; blue,204},
    font=\small, rounded corners=3pt, minimum width=2cm, align=center},
  paramHalf/.style={
    draw, rectangle,
    minimum width=2cm, align=center, font=\small, rounded corners=3pt,
    path picture={
      \begin{scope}[sharp corners]
        \fill[gray!20]
          ($ (path picture bounding box.south west)!0.5!(path picture bounding box.north west) $)
          rectangle
          (path picture bounding box.north east);
        \fill[draw=none, fill=white]
          (path picture bounding box.south west)
          rectangle
          ($ (path picture bounding box.south east)!0.5!(path picture bounding box.north east) $);
      \end{scope}
    }
  }
}

\tcolorboxenvironment{mytheorem}{
  colframe=orange,
  sharp corners=west,
  colback=orange!10,
  rightrule=0pt,
  toprule=0pt,
  bottomrule=0pt,
  top=0pt,
  right=0pt,
  bottom=0pt,
}

\tcolorboxenvironment{mylemma}{
  colframe=yellow,
  sharp corners=west,
  colback=yellow!10,
  rightrule=0pt,
  toprule=0pt,
  bottomrule=0pt,
  top=0pt,
  right=0pt,
  bottom=0pt,
}

\renewcommand{\epsilon}{\varepsilon}
\renewcommand{\phi}{\varphi}

\newcommand{\pw}{\mathrm{pw}}
\newcommand{\ctw}{\mathrm{ctw}}

\newcommand{\dist}{\mathrm{dist}}

\newcommand{\XNLP}{\ensuremath{\mathsf{XNLP}}}

\newcommand{\Omc}{\ensuremath{\mathcal{O}}\xspace}

\newcommand{\Oof}{\Omc}

\newcommand{\N}{\mathbb{N}}
\newcommand{\Z}{\mathbb{Z}}

\newcommand{\PathDiscovery}{\textnormal{\textsc{Path Discovery}}\xspace}
\newcommand{\ShortestPathDiscovery}{\textnormal{\textsc{Shortest Path Discovery}}\xspace}
\newcommand{\BothPathDiscovery}{\textnormal{\textsc{(Shortest) Path Discovery}}\xspace}

\newcommand{\indeg}{\inw}

\Crefname{corollary}{Corollary}{Corollaries}
\Crefname{lemma}{Lemma}{Lemmas}
\Crefname{section}{Section}{Sections}
\Crefname{theorem}{Theorem}{Theorems}
\Crefname{proposition}{Proposition}{Propositions}
\Crefname{observation}{Observation}{Observations}
\Crefname{remark}{Remark}{Remarks}

\title{Separating Feasibility and Movement in Solution Discovery: The Case of Path Discovery}

\titlerunning{Separating Feasibility and Movement in Solution Discovery}

\author{Hanno von Bergen}
{Universität Hamburg, Germany}
{hanno.von.bergen@uni-hamburg.de}
{}
{}

\author{Larissa Fastenau}
{University of Bremen, Germany}
{larissa.fast.15@gmail.com}
{}
{}

\author{Enna Gerhard}
{University of Bremen, Germany}
{gerhard@uni-bremen.de}
{https://orcid.org/0000-0002-7767-6637}
{}

\author{Nicola Lorenz}
{Universität Hamburg, Germany}
{nicola.lorenz@uni-hamburg.de}
{https://orcid.org/0009-0000-1991-1523}
{}

\author{Stephanie Maaz}
{University of Waterloo, Canada}
{smaaz@uwaterloo.ca}
{https://orcid.org/0000-0001-7188-8834}
{}

\author{Amer E. Mouawad}
{American University of Beirut, Lebanon}
{aa368@aub.edu.lb}
{https://orcid.org/0000-0003-2481-4968}
{}

\author{Roman Rabinovich}
{Technische Universität Berlin, Germany}
{roman.rabinovich@tu-berlin.de}
{}
{}

\author{Nicole Schirrmacher}
{University of Bremen, Germany}
{schirrmacher@uni-bremen.de}
{https://orcid.org/0000-0002-1740-7478}
{}

\author{Daniel Schmand}
{University of Bremen, Germany}
{schmand@uni-bremen.de}
{https://orcid.org/0000-0001-7776-3426}
{}

\author{Sebastian Siebertz}
{University of Bremen, Germany}
{siebertz@uni-bremen.de}
{https://orcid.org/0000-0002-6347-1198}
{}

\author{Mai Trinh}
{Universität Hamburg, Germany}
{mai.trinh@uni-hamburg.de}
{}
{}

\authorrunning{S. Siebertz et al.}

\Copyright{Hanno von Bergen, Larissa Fastenau, Enna Gerhard, Nicola Lorenz, Stephanie Maaz, Amer Mouawad, Roman Rabinovich, Nicole Schirrmacher, Daniel Schmand, Sebastian Siebertz, and Mai Trinh}

\ccsdesc[500]{Theory of computation~Parameterized complexity and exact algorithms}
\ccsdesc[300]{Theory of computation~Routing and network design problems}
\ccsdesc[300]{Mathematics of computing~Graph algorithms} 

\keywords{solution discovery, shortest path discovery, token sliding, parameterized complexity}

\category{}
\relatedversion{}
\supplement{}

\acknowledgements{}

\EventEditors{Editor 1 and Editor 2}
\EventNoEds{2}
\EventLongTitle{}
\EventShortTitle{MFCS 2026}
\EventAcronym{MFCS}
\EventYear{2026}
\EventDate{Month 00--00, 2026}
\EventLocation{City, Country}
\EventLogo{}
\SeriesVolume{XX}
\ArticleNo{XX}

\begin{document}

\maketitle

\begin{abstract}
We study \emph{solution discovery}, where the goal is to obtain a feasible solution to a problem from an initial configuration by a bounded sequence of local moves. 
In many applications, however, the graph that defines which vertex sets are feasible is not the same as the graph that governs how tokens, agents, or resources may move. 
Existing models such as token sliding and token jumping typically do not distinguish the problem graph and the movement graph. 
Motivated by this mismatch, we introduce a directed weighted \emph{two-graph model} that cleanly separates feasibility from movement. 
A problem graph specifies the desired combinatorial objects, while a movement graph specifies admissible relocations and their costs. 
This yields a flexible framework that captures asymmetry, heterogeneous movement constraints, and weighted transitions, while subsuming classical discovery models as special cases.

We investigate this model through \textsc{Path Discovery} and \textsc{Shortest Path Discovery}, where the task is to realize a vertex set containing an $s$-$t$-path or a shortest $s$-$t$-path in the problem graph. 
These problems are particularly natural in applications, since directed and weighted shortest paths are among the most fundamental algorithmic primitives. 
At the same time, previous work has already shown that discovery can be computationally hard even when the underlying optimization problem is easy. 
Our results show that this phenomenon persists, and becomes especially rich, in the two-graph setting. 
We obtain a detailed complexity picture, identifying tractable cases as well as strong hardness results.
\end{abstract}

\section{Introduction}
Combinatorial reconfiguration studies the structure of the solution space of a
combinatorial problem under local transformation rules.
Typically, one is given two feasible solutions and asks whether one can transform one
into the other by a sequence of small local moves while maintaining feasibility
throughout.
Over the past years, combinatorial reconfiguration has developed into a rich area at
the interface of algorithms, complexity theory, graph theory, and optimization; we
refer to the surveys by van den Heuvel and Nishimura for background and further
references~\cite{Nishimura2018, Heuvel2013}.
Besides its intrinsic mathematical appeal, reconfiguration is motivated by dynamic and
incremental settings in which a system state must be adapted while maintaining feasibility.

Recently, Fellows et al.\ introduced the framework of
\emph{solution discovery via reconfiguration} as a framework inspired by combinatorial reconfiguration~\cite{fellows2023solution}.
Here, one is not given a target solution in advance.
Instead, one starts from an initial configuration and asks whether it can be modified
by a bounded sequence of local moves so as to obtain \emph{some} feasible solution.
This viewpoint is natural in settings where a system is already partially deployed,
resources are already placed, or a previously valid state has become infeasible after
local changes.
The framework connects reconfiguration with local search, repair, reoptimization, and
dynamic adaptation, and it has already led to a systematic study of discovery variants
of classical graph problems~\cite{bousquet2025algorithmic, fellows2023solution,GroblerEtAlICALP24, GroblerMMNRS242,saito2026solution}.


A conceptual limitation of the classical token sliding and token jumping models is
that they do not distinguish two structures that are often different in applications:
the structure that defines \emph{feasibility} of a target solution, and the structure
that governs \emph{movement}.
For example, consider the task of setting up a wireless communication system by
positioning radio-equipped vehicles such that all areas of interest are covered.
The coverage depends on the topology of the area, which determines which sets of
positions yield a feasible solution.
In contrast, the movement of the vehicles is constrained by the road network, which
imposes its own structure, directionality, and costs.
Thus, the structure governing feasibility and the structure governing movement are
inherently different.
This phenomenon is common in applications where resources must be arranged to satisfy
structural constraints (such as routes, assignments, or coverage), but can only be
moved along a different underlying network with its own constraints.
Such examples suggest that feasibility and movement should be modeled separately.


Motivated by this observation, we introduce a directed weighted \emph{two-graph model}
for solution discovery.
In this model, we are given a directed weighted \emph{problem graph} $G$, which specifies the target
objects whose discovery we are interested in, and a directed weighted
\emph{movement graph}~$M$, which specifies admissible token moves and their costs.
An instance further contains an initial token configuration and a movement budget.
The aim is to move tokens in $M$ so that the finally occupied vertices realize a
desired structure in $G$.
Classical models arise as special cases of our framework.
The token sliding model is obtained when $M$ coincides with $G$ (or its bi-directed
version), so that tokens may only move along edges of the problem graph.
The token jumping model is obtained when $M$ is the complete bi-directed graph with
unit costs, allowing tokens to move freely between any pair of vertices.
Thus, our model strictly generalizes the standard ones while retaining them as natural
special cases.

We study this new paradigm through the lens of \textsc{Path Discovery} and
\textsc{Shortest Path Discovery}.
In \textsc{Path Discovery}, the goal is to realize a vertex set that contains some directed
$s$-$t$-path in the problem graph $G$; in \textsc{Shortest Path Discovery}, the goal
is to realize a vertex set containing a shortest directed $s$-$t$-path in $G$.
This choice is natural for two reasons.
First, directed and weighted shortest paths are among the most fundamental problems in
classical algorithmics, with countless applications and a long history of efficient algorithms.
Second, they provide a clean way to isolate the effect of the discovery process itself:
while the underlying shortest-path problem is polynomial-time solvable, its discovery variant
is already computationally hard even in the standard undirected token-sliding
setting~\cite{GroblerEtAlICALP24}.
Thus, \textsc{Shortest Path Discovery} is a particularly well-suited benchmark for
understanding how the separation of feasibility and movement affects the complexity landscape.


\textsc{Shortest Path Reconfiguration} in undirected graphs has already been studied in the classical reconfiguration setting, where
both the initial and the target shortest path are given and every intermediate state
must remain feasible~\cite{KaminskiMedvedevMilanic2012}, and \textsc{Shortest Path Discovery} in undirected graphs has been studied in~\cite{GroblerEtAlICALP24}. 
In our study it quickly turns out that \textsc{Path Discovery} is the more fundamental problem in the two-graph model, as \textsc{Shortest Path Discovery} can in most cases be reduced to it by restricting the edge set of the problem graph to those edges that lie on shortest $s$-$t$-paths, while keeping the vertex set unchanged.
In particular, whenever the considered parameter is not increased by passing to this
shortest-path subgraph, our tractability results for \textsc{Path Discovery} also apply to
\textsc{Shortest Path Discovery}.


Our results show that the separation of feasibility and movement leads to a robust and
expressive framework while preserving substantial algorithmic structure.
We obtain a polynomial-time algorithm for the token jumping case, and we show that \textsc{(Shortest) Path Discovery} is
fixed-parameter tractable parameterized by the number~$k$ of tokens.
\textsc{(Shortest) Path Discovery} is fixed-parameter tractable for the
structural parameter feedback edge set number. 
We show that \textsc{(Shortest) Path Discovery} is fixed-parameter tractable when parameterized by any parameter that functionally bounds the number of vertices on all simple directed $s$-$t$-paths in~$G$, e.g.\ the parameter vertex integrity or treedepth of the underlying undirected graph, or for parameter ``distance from $s$ to $t$'' for \textsc{Shortest Path Discovery} in case all edge weights are positive. 
We also obtain \textsf{XP}-algorithms on graph classes of bounded treewidth when treewidth is
measured in the undirected union of the problem graph and the movement graph.

On the hardness side, we show that the problem remains difficult already in very restricted settings.
In particular, \textsc{Path Discovery} is para-$\mathsf{NP}$-hard for parameter distance from $s$ to $t$ even
when all edge weights in the problem graph are positive, while allowing edge weight
$0$ yields para-$\mathsf{NP}$-hardness for both \textsc{Path Discovery} and
\textsc{Shortest Path Discovery} with respect to the parameter distance from $s$ to $t$.
Similarly, if zero-weight edges are allowed in the movement graph, then both problems
become para-$\mathsf{NP}$-hard for the budget parameter $b$.
In the classical undirected unweighted setting, we also obtain new hardness results,
including $\mathsf{NP}$-hardness on planar graphs and \XNLP-hardness for cutwidth (which implies $\mathsf{W}[i]$-hardness for all $i$ for these parameters). 
Since cutwidth bounds treewidth, the cutwidth hardness also yields hardness for
treewidth, and therefore suggests that our XP-algorithm for treewidth cannot be
improved to an \FPT-algorithm under standard assumptions.
Moreover, the same ideas show hardness for the parameter directed feedback vertex set.

Taken together, our results show that separating feasibility and movement is both
natural and algorithmically fruitful.
The model captures asymmetry, heterogeneity, and weighted movement costs that are
unavoidable in realistic applications, while opening a broad new landscape of
complexity questions at the interface of shortest paths, solution discovery, and
structural parameterized complexity.

\shortlongversion{Due to space constraints, we only sketch some of the results in the submitted abstract. The full version can be found in the appendix.}{}

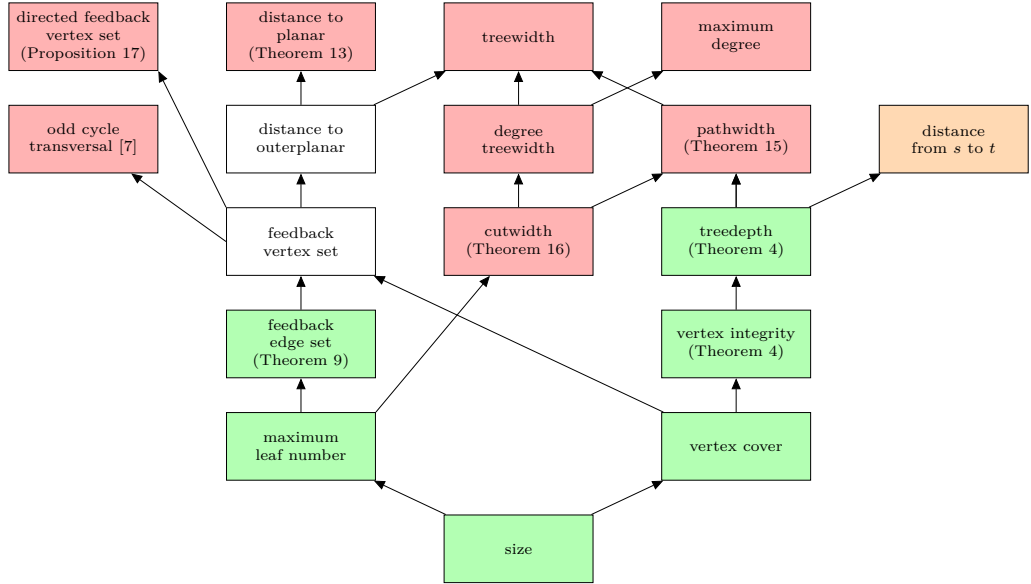
\begin{figure}
  \centering
  \scalebox{.9}{
    \tikzset{
  box/.style={
      draw,
      rectangle,
      minimum width=22mm,
      minimum height=10mm,
      align=center,
      font=\fontsize{6.5pt}{7.5pt}\selectfont,
      inner sep=1pt
    },
  arrow/.style={
      {-Stealth[inset=0pt, angle=45:5pt]}
    }
}

\begin{tikzpicture}[node distance=10mm]
  \node[box, fill=green!30] (size) {size};
  \node[box, fill=green!30, above left=5mm and 10mm of size] (leaf) {maximum \\leaf number};
  \node[box, fill=green!30, above right=5mm and 10mm of size] (vc) {vertex cover};
  \node[box, fill=green!30, above=5mm of leaf] (fes) {feedback \\edge set \\(\Cref{thm:path-fes})};
  \node[box, above=5mm of fes] (fvs) {feedback \\vertex set};
  \node[box, fill=red!30, above left=5mm and 10mm of fvs] (oct) {odd cycle \\transversal~\cite{GroblerEtAlICALP24}};
  \node[box, fill=red!30, above=5mm of oct] (dfvs) {directed feedback \\vertex set \\(\Cref{prop:dfvs})};
  \node[box, fill=red!30, right= of fvs] (cw) {cutwidth\\(\Cref{thm:xnlp-hardness-cutwidth})};
  \node[box, above=5mm of fvs] (distouter) {distance to \\outerplanar};
  \node[box, fill=red!30, above=5mm of distouter] (planar) {distance to \\planar \\(\Cref{thm:planar-hardness})};
  \node[box, fill=red!30, above=5mm of cw] (dtw) {degree \\treewidth};
  \node[box, fill=red!30, right= of dtw] (pw) {pathwidth \\(\Cref{thm:xnlp-hardness-pathwidth})};
  \node[box, fill=red!30, above=5mm of dtw] (tw) {treewidth};
  \node[box, fill=red!30, above=5mm of pw] (deg) {maximum \\degree};
  \node[box, fill=green!30, above=5mm of vc] (vi) {vertex integrity \\(\Cref{thm:param-induced-path})};
  \node[box, fill=green!30, above=5mm of vi] (td) {treedepth \\(\Cref{thm:param-induced-path})};
  \node[box, fill=orange!30, above right=5mm and 10mm of td] (diam) {distance\\ from $s$ to $t$};

  \draw[arrow] (size) -- (leaf);
  \draw[arrow] (size) -- (vc);
  \draw[arrow] (leaf.north east) -- (cw);
  \draw[arrow] (leaf) -- (fes);
  \draw[arrow] (vc) -- (vi);
  \draw[arrow] (vc) -- (fvs);
  \draw[arrow] (cw) -- (dtw);
  \draw[arrow] (cw) -- (pw);
  \draw[arrow] (fes) -- (fvs);
  \draw[arrow] (vi) -- (td);
  \draw[arrow] (dtw) -- (deg);
  \draw[arrow] (dtw) -- (tw);
  \draw[arrow] (pw) -- (tw);
  \draw[arrow] (fvs) -- (distouter);
  \draw[arrow] (fvs.west) -- (oct);
  \draw[arrow] (fvs.north west) -- (dfvs.south east);
  \draw[arrow] (td) -- (pw);
  \draw[arrow] (td) -- (diam);
  \draw[arrow] (td) -- (pw);
  \draw[arrow] (distouter) -- (planar);
  \draw[arrow] (distouter) -- (tw);
\end{tikzpicture}
  }
  \caption{Parameters considered in this paper. \textsc{Shortest Path Discovery} is tractable when parameterized by distance from $s$ to $t$ (in particular by parameter diameter) when all edge weights are positive. The problem is hard when zero weights are allowed. All parameters except for directed feedback vertex set and distance refer to the measure on the underlying undirected movement graph.}
  \label{fig:parameters}
\end{figure}

\section{Preliminaries}
\label{sec:preliminaries}

For an integer $n\in \N$, we write $[n]\coloneqq \{1,\dots,n\}$.
For a graph $G$, we denote its vertex and edge sets by $V(G)$ and $E(G)\subseteq V(G)\times V(G)$, respectively.
We additionally assume a weight function $w_G\colon E(G)\to \Z$ assigning an integer weight
to each edge. 
For $s,t\in V(G)$, an \emph{$s$-$t$ path} is a sequence
$P=(v_0, v_1,\ldots,v_\ell)$ of pairwise distinct vertices with $v_0=s$, $v_\ell=t$, and $e_i=(v_{i-1},v_i)\in E(G)$ for all $i\in[\ell]$.
The \emph{weight} of a path $P$ is
$w_G(P)\;\coloneqq\; \sum_{i=1}^{\ell} w_G(e_i)$,
and when~$G$ is clear from the context, we simply write $w(e)$ and $w(P)$.
The \emph{distance} from $u$ to $v$ in $G$, denoted $\dist_G(u,v)$, is the minimum weight of a $u$-$v$ path
(if no such path exists, we set $\dist_G(s,t)=\infty$).
A \emph{shortest $s$-$t$ path} is an $s$-$t$ path $P$ with $w_G(P)=\dist_G(s,t)$.
Note that distance is well-defined as we exclude negative cycles and that a shortest walk will always contain a shortest path of the same weight.

Whenever convenient, we view an undirected graph $H$ as a directed graph by replacing
each undirected edge $\{x,y\}\in E(H)$ with two directed edges $(x,y)$ and $(y,x)$ of the same weight.
Conversely, if $G$ is a graph, we write $U$ for the underlying undirected graph with $V(U)=V(G)$ and $\{x,y\}\in E(U)$ if and only if $(x,y)\in E(G)$ or $(y,x)\in E(G)$. 

We work with two directed weighted graphs on a common vertex set:
the \emph{problem graph}~$G$, which specifies feasibility/optimality of target objects,
and the \emph{movement graph} $M$, which specifies admissible token moves and their costs.
Problem-graph edge weights may be arbitrary integers, but we assume that $G$
contains no negative cycle. Movement costs are nonnegative integers, 
\(w_M:E(M)\to \mathbb Z_{\geq 0}\).

The restriction on the weights of $G$ is standard, as it ensures e.g.\ that distances are well defined. 
The reason for the restriction on the weights of $M$ is that with negative movement edges, unused tokens could be moved along negative-cost paths and then lifted at the end, invalidating some of our observations. 

A \emph{configuration} is a multiset $S\subseteq V(M)$. Vertices in $S$ are called \emph{occupied} by tokens and vertices in
$V(M)\setminus S$ are \emph{free}.
A \emph{(token) move} in configuration $S$ consists of choosing a token vertex $x\in S$ and an edge $(x,y)\in E(M)$, and updating the configuration to
\[
S' \;=\; (S\setminus\{x\}) \cup \{y\} \text{ (as a multiset)}.
\]
The cost of this move is $w_M(x,y)$.
A \emph{discovery sequence} from $S$ to $S'$ is a sequence of configurations
$S=S_0,S_1,\dots,S_r=S'$ such that $S_{i+1}$ is obtained from $S_i$ by a single move for each $i\in\{0,\dots,r-1\}$.
Its total cost is $\sum_{i=0}^{r-1} w_M(e_i)$ where $e_i$ is the movement edge used in the $i$th move.
We say that $S'$ is \emph{reachable from $S$ within budget $b$} if there exists a discovery sequence from $S$
to $S'$ of total cost at most $b$.


An instance of \textsc{(Shortest) Path Discovery} consists of a problem graph $G$, a movement graph $M$
(on the same vertex set), two distinct terminals $s,t\in V(G)$, an initial configuration $S\subseteq V(M)$,
and a budget $b\in \Z$.
The question is whether there exists a configuration $T\subseteq V(M)$ reachable from $S$ within budget $b$ such that $T$ contains the vertex set of some (shortest) $s$-$t$-path in~$G$, that is, there exists a (shortest) $s$-$t$-path $P$ in $G$ with $V(P)\subseteq T$.

We emphasize that we do \emph{not} require $T=V(P)$ as it was the case in~\cite{GroblerEtAlICALP24} for the unweighted case. The motivation for this requirement in~\cite{GroblerEtAlICALP24} is that in unweighted graphs, the number of vertices on a shortest path is uniquely determined, and we hence know exactly how many tokens are required for a shortest path. This is not the case in the weighted model.
Equivalently, we may allow to lift tokens in the end at zero cost, and we will take this intuitive view, when we would like to say that the final token configuration $T$ highlights the discovered path $P$.  
We also remark that the stacking of multiple tokens on single vertices was not allowed in previous work, but we allow this just for convenience. If during a sequence two tokens would temporarily ``cross'' or meet, one may simply swap their roles: we may simply let the other token move on and let the first token take its place. 


\medskip
We assume familiarity with parameterized complexity, and refer e.g.\ to~\cite{cygan2015parameterized} for more background. 
A problem is \emph{fixed-parameter tractable} (\FPT) with respect to a parameter $p$ if it can be solved in time
$f(p)\cdot n^{O(1)}$ for some computable function $f$.
It is \XP\ with respect to $p$ if it can be solved in time $n^{f(p)}$.
We use $\W[i]$ for the standard $\W$-hierarchy.
We also use the class \textsf{XNLP} and \textsf{XNLP}-hardness in the sense of parameterized logspace computation;
in particular, \textsf{XNLP}-hardness implies $\W[i]$-hardness for every~$i$.

\section{The tractable cases}

In this section, we provide our positive results. 

\subsection{Parameter number of tokens}

In this section, we prove that \textsc{Path Discovery} and \textsc{Shortest Path Discovery} are fixed-parameter tractable when parameterized by the number of tokens $k$. 

We begin with a remark that allows us to focus on \textsc{Path Discovery} (in most cases) when proving positive results. 

\begin{obs}\label{obs:SPD=PD}
    Let $G$ be a graph and let
\[
d_s(v):=\dist_G(s,v),\qquad
d_t(v):=\dist_G(v,t),\qquad
D:=\dist_G(s,t).
\]
If $D=\dist_G(s,t)=\infty$, then the instance is immediately negative.
Otherwise, define the subgraph $G^\star$ of $G$ by $V(G^\star)=V(G)$,
and
\[
 E(G^\star)=
 \{(u,v)\in E(G)\mid d_s(u)<\infty,\ d_t(v)<\infty,\ 
 d_s(u)+w_G(u,v)+d_t(v)=D\}.
\]
Then the directed $s$-$t$-paths in $G^\star$ are exactly the shortest directed $s$-$t$-paths in $G$.
\ShortestPathDiscovery on $G$ is equivalent to \PathDiscovery on $G^\star$.
\end{obs}

\begin{theorem}
\label{thm:two-graph-discovery-fpt-k}
\BothPathDiscovery is fixed-parameter tractable parameterized by~$k$.
\end{theorem}

\shortlongversion{
    \begin{proof}[Proof sketch]
By \Cref{obs:SPD=PD}, it suffices to consider \textsc{Path Discovery}.  Let \(k\) be the number of tokens.  We guess the number \(r\le k\) of tokens that will occupy the discovered path, together with these tokens in their order along the path.  For a fixed ordered list \(Y=(y_1,\ldots,y_r)\), we build a layered auxiliary digraph whose \(i\)-th layer contains a copy \((v,i)\) of every vertex \(v\in V(G)\); an arc from layer \(i\) to layer \(i+1\) represents choosing a problem-graph edge of \(G\), and its weight is the movement cost of token \(y_{i+1}\) to the new vertex.  Thus, an \(\alpha\)-\(\omega\) path in this auxiliary graph represents an \(s\)-\(t\)-walk in \(G\), with total weight equal to the cost of moving the guessed tokens to the chosen vertices.  The only remaining issue is to ensure that the projected walk is simple.  We do this by the color-coding technique of Alon, Yuster, and Zwick~\cite{AlonYusterZwick1995}: using a family of \(r\)-perfect hash functions, we try colorings of \(V(G)\) such that every possible simple \(r\)-vertex path is colorful for some coloring.  For each guessed token order and coloring, we solve the resulting weighted colorful path problem in time \(2^{\mathcal O(r)}|V|^{\mathcal O(1)}\).  If a colorful auxiliary path of weight at most \(b\) exists, it gives a valid discovered \(s\)-\(t\)-path; conversely, any valid solution is found when we guess its token order and a coloring injective on its path vertices.  Since there are at most \(\sum_{r=2}^k k!/(k-r)! \le 2^k k!\) token-order guesses and \(r\le k\), the total running time is \(f(k)\cdot |V|^{\mathcal O(1)}\), proving fixed-parameter tractability.
\end{proof}
}{
\begin{proof}
We consider \textsc{Path Discovery}; the result for \textsc{Shortest Path Discovery} follows from \Cref{obs:SPD=PD}.
Let $G=(V,E(G),w_G)$ and $M=(V,E(M),w_M)$ be directed weighted graphs on the same vertex set $V$, where
$w_G\colon E(G)\to \mathbb Z$ and $w_M\colon E(M)\to \mathbb Z$.

Let $S=\{x_1,\dots,x_k\}\subseteq V$ be the multiset of initial token positions, let $s,t\in V$, and let~$b\in \mathbb Z$.


A solution consists of an integer $r\in\{2,\dots,k\}$, pairwise distinct tokens $y_1,\dots,y_r\in S$,
pairwise distinct vertices $v_1,\dots,v_r\in V$,
such that $v_1=s$, $v_r=t$, $(v_i,v_{i+1})\in E(G)$ for all $i\in\{1,\dots,r-1\}$, and
$\sum_{i=1}^r \dist_M(y_i,v_i)\le b$.
Since unused tokens may be lifted in the end, and there are no negative weights in the movement graph, only the tokens $y_1,\dots,y_r$ matter for the solution.

We guess the participating tokens and their order on the final path.
That is, we guess an ordered list
$Y=(y_1,\dots,y_r)$
of pairwise distinct tokens from $S$, for some $r\in\{2,\dots,k\}$.
The number of such guesses is
$
\sum_{r=2}^k \frac{k!}{(k-r)!}\le 2^k\cdot k!$.
Fix one such guess $Y=(y_1,\dots,y_r)$.

We now construct a layered weighted digraph $U_Y$, called the \emph{unraveling graph}, 
with vertex set 
$
V(U_Y)=\{(v,i)\mid v\in V,\ i\in\{1,\dots,r\}\}$.
The intended meaning of $(v,i)$ is that the $i$-th vertex of the final path is $v$, occupied by token $y_i$.

We add a fresh source $\alpha$ and a fresh sink $\omega$.
We add the arc $\alpha \to (s,1)$
with weight $\dist_M(y_1,s)$, and for every $i\in\{1,\dots,r-1\}$ and every arc $(u,v)\in E(G)$, we add the arc $(u,i)\to (v,i+1)$
with weight $\dist_M(y_{i+1},v)$.
Finally, we add the arc $(t,r)\to \omega$
with weight~$0$.

Every directed $\alpha$-$\omega$ path in $U_Y$ corresponds to a directed walk $v_1=s,v_2,\dots,v_r=t$
in~$G$, and its total weight is exactly $\sum_{i=1}^r \dist_M(y_i,v_i)$.
The problem is that such a walk need not be simple, because the same original vertex $v\in V$ may appear in several layers as $(v,i)$ and~$(v,j)$.

To enforce simplicity without storing all already used vertices explicitly, which would result in a running time $|V|^{\Oof(k)}$, we use the color-coding technique of Alon, Yuster, and Zwick~\cite{AlonYusterZwick1995}.
Fix $r$, and let $\mathcal F_r$ be a family of functions $\lambda\colon V\to [r]$
such that for every set $X\subseteq V$ of size $r$, some $\lambda\in\mathcal F_r$ is injective on $X$. 
We interpret the assignment $\lambda$ as a coloring of $V$. 
Such a family of size $2^{\Oof(r)}\log |V|$ can be constructed in time
$2^{\Oof(r)}\cdot |V|^{\Oof(1)}$~\cite{AlonYusterZwick1995}.

For a fixed $\lambda\in\mathcal F_r$, we call an $\alpha$-$\omega$ path in $U_Y$ \emph{colorful} if the corresponding vertices $v_1,\dots,v_r$ of~$G$ have pairwise distinct colors under $\lambda$.
Every colorful such path corresponds to a simple directed $s$-$t$-path in $G$.

Hence, for fixed $Y$ and $\lambda$, it suffices to determine whether $U_Y$ contains a minimum-weight colorful $\alpha$-$\omega$ path using exactly one vertex from each layer.
This is precisely an instance of the weighted colorful $r$-path problem on a directed graph, and it can be solved in time $2^{\Oof(r)}\cdot |V|^{\Oof(1)}$ by the color-coding machinery~\cite{AlonYusterZwick1995}.

We accept if for some $r$, some ordered token list $Y$, and some coloring $\lambda\in\mathcal F_r$, the minimum weight of a colorful $\alpha$-$\omega$ path in $U_Y$ is at most $b$.

\medskip
We prove correctness.
Assume first that the algorithm accepts.
Then for some $r$, some token order $Y=(y_1,\dots,y_r)$, and some $\lambda\in\mathcal F_r$, there exists a colorful $\alpha$-$\omega$ path in $U_Y$ of weight at most~$b$.
Let $\alpha,(v_1,1),(v_2,2),\dots,(v_r,r),\omega$
be this path.
By construction of $U_Y$, we have $v_1=s$, $v_r=t$, and $(v_i,v_{i+1})\in E(G)$ for all $i\in\{1,\dots,r-1\}$.
Since the path is colorful, the vertices $v_1,\dots,v_r$ are pairwise distinct.
Hence, they form a directed simple $s$-$t$-path in $G$.
Moreover, the weight of the unraveling path is exactly
$\sum_{i=1}^r \dist_M(y_i,v_i)\le b$.
Thus moving token $y_i$ to $v_i$ for each $i$ yields a valid discovery solution.

Conversely, suppose the instance is a yes-instance.
Then there exist $r$, pairwise distinct tokens $y_1,\dots,y_r\in S$, and a directed simple $s$-$t$-path $v_1=s,v_2,\dots,v_r=t$ in $G$ such that
$\sum_{i=1}^r \dist_M(y_i,v_i)\le b$.
When the algorithm guesses the ordered token list $Y=(y_1,\dots,y_r)$, it considers exactly this order.
Since $\{v_1,\dots,v_r\}$ has size $r$, there is some $\lambda\in\mathcal F_r$ that is injective on this set.
Therefore, the corresponding $\alpha$-$\omega$ path
$\alpha,(v_1,1),\dots,(v_r,r),\omega$
in $U_Y$ is colorful, and has exactly this weight. 
Hence, the algorithm accepts.

\medskip
For the running time, for each $r$ there are at most $k!/(k-r)!$ choices of $Y$, and for each such choice we consider a family $\mathcal F_r$ of size $2^{\Oof(r)}\log |V|$ computable in time $2^{\Oof(r)}\cdot |V|^{\Oof(1)}$.
For each pair $(Y,\lambda)$, the weighted colorful $r$-path problem on $U_Y$ is solved in time $2^{\Oof(r)}\cdot |V|^{\Oof(1)}$.
Thus, the total running time is
\[
\sum_{r=2}^k \frac{k!}{(k-r)!}\cdot 2^{\Oof(r)}\cdot |V|^{\Oof(1)}
\le
2^k\cdot k!\cdot 2^{\Oof(k)}\cdot |V|^{\Oof(1)},
\]
which is of the form $f(k)\cdot |V|^{\Oof(1)}$.
Hence, \textsc{Path Discovery} is fixed-parameter tractable parameterized by $k$.
\end{proof}}

\subsection{Tractability of token jumping}

By using that in the token jumping model, we have unit costs to move a token anywhere, we do not have to guess the order of tokens as in \Cref{thm:two-graph-discovery-fpt-k}.
Since the target object is a simple $s$-$t$-path and every extra repeated cycle in a
projected walk can be deleted without increasing the number of newly occupied vertices,
we can avoid the color-coding machinery in the token-jumping case. 

\begin{theorem}
\label{cor:token-jumping-polytime}
\BothPathDiscovery in the token jumping model is solvable in polynomial time.
\end{theorem}

\shortlongversion{}{
\begin{proof}
    We consider \textsc{Path Discovery}, as the result for \textsc{Shortest Path Discovery} follows from \Cref{obs:SPD=PD}. 
In the token jumping model, the movement graph $M$ is the bidirected clique with unit edge weights.
Therefore, if $P$ is a directed simple $s$-$t$-path, then the minimum cost of moving tokens so as to occupy exactly the vertices of $P$ is
$\mu(P)=|V(P)\setminus S|$,
provided that $|V(P)|\leq |S|=k$; if $|V(P)|>k$, then $P$ cannot be realized since there are only $k$ tokens.

Thus, the identity and order of the participating tokens are irrelevant: only the set of vertices of the final path matters.

As above, we construct a layered directed graph $U$ with vertex set
$V(U)=\{(v,i)\mid v\in V,$ $i\in\{1,\dots,k\}\}\cup\{\alpha,\omega\}$,
where $\alpha$ is a fresh source and $\omega$ a fresh sink.
Again, the intended meaning of $(v,i)$ is that $v$ is chosen as the $i$-th vertex of the final path.

We add the arc $\alpha\to (s,1)$
and the arc $(t,i)\to \omega$ for every $i\in\{1,\dots,k\}$.
For every $i\in\{1,\dots,k-1\}$ and every arc $(u,v)\in E(G)$, we add the arc $(u,i)\to (v,i+1)$.
We assign weight
\[
c(v):=
\begin{cases}
0 & \text{if } v\in S,\\
1 & \text{if } v\notin S.
\end{cases}
\]
to each chosen path vertex $v$.
Accordingly, we give the arc $\alpha\to (s,1)$ weight $c(s)$, every arc $(u,i)\to (v,i+1)$
weight $c(v)$, and every arc $(t,i)\to\omega$ weight $0$.

Then every directed $\alpha$-$\omega$ path in $U$, $\alpha,(v_1,1),(v_2,2),\dots,(v_r,r),\omega$,
corresponds to a directed $s$-$t$-walk
$v_1=s,v_2,\dots,v_r=t$
in $G$ with $r\le k$, and the total weight of the path in $U$ is exactly
$\sum_{j=1}^r c(v_j)$.
This is the number of selected vertices that are not initially occupied by a token.

Since all arc weights in $U$ are non-negative, if the projected walk in $G$ repeats a vertex, then deleting the corresponding cycle does not increase the total weight in $U$.
Hence, there exists an optimal $\alpha$-$\omega$ path in $U$ whose projection to $G$ is a simple $s$-$t$-path.
Therefore, the minimum weight of an $\alpha$-$\omega$ path in $U$ is exactly the minimum movement cost of a realizable $s$-$t$-path in $G$.

Thus, \textsc{Path Discovery} can be solved by a shortest path computation in the layered graph $U$, and hence in polynomial time.
\end{proof}}

\subsection{Parameter bounded solution size}

The color-coding argument from the previous subsection does not only apply when the
number of tokens is bounded. 
The difficulty is that we do not know in advance which of the tokens of $S$ will be used on the
final path, so we cannot simply restrict to a bounded number of tokens. To resolve this, we again apply 
color-coding, this time simultaneously to the path vertices and to the used tokens.

\begin{theorem}\label{thm:param-induced-path}
    \BothPathDiscovery is fixed-parameter tractable when parameterized by any parameter
    that functionally bounds the number of vertices on all simple directed $s$-$t$-paths in~$G$.
\end{theorem}

\shortlongversion{}{
\begin{proof}
Let $(G,M,s,t,S,b)$ be an instance of \textsc{Path Discovery}, and let $r$ be an upper bound on the
number of vertices on any simple directed $s$-$t$-path in $G$.
We show fixed-parameter tractability parameterized by $r$. 
The result for \textsc{Shortest Path Discovery} follows by \Cref{obs:SPD=PD}.

\medskip 

Fix an integer $\ell\in\{2,\dots,r\}$. 
We seek a directed simple $s$-$t$-path on exactly $\ell$ vertices.
Again, using the color-coding technique, we compute a family~$\mathcal F_\ell$ of functions
$\lambda_V \colon V(G)\to [\ell]$
of size $2^{\Oof(\ell)}\log |V(G)|$ such that for every set $X\subseteq V(G)$ of size $\ell$, some
$\lambda_V\in\mathcal F_\ell$ is injective on $X$. Likewise, we compute a family $\mathcal H_\ell$
of functions
$\lambda_T \colon S\to [\ell]$
of size $2^{\Oof(\ell)}\log |S|$ such that for every set $Y\subseteq S$ of size $\ell$, some
$\lambda_T\in\mathcal H_\ell$ is injective on $Y$.

Fix $\lambda_V\in\mathcal F_\ell$ and $\lambda_T\in\mathcal H_\ell$.
For every color $i\in[\ell]$ and every vertex $v\in V(G)$, define
$c_i(v):=\min\{\dist_M(x,v)\mid x\in S,\ \lambda_T(x)=i\}$,
where $c_i(v):=+\infty$ if no token of color $i$ can reach $v$.
Intuitively, if the $j$-th vertex of the final path receives vertex-color $i$, then we pay the
minimum cost of assigning to it a token of token-color $i$.

We now guess a permutation $\pi$ of $[\ell]$. 
The intended meaning is that the $j$-th vertex of the path has vertex-color $\pi(j)$, and
is occupied by a token of token-color $\pi(j)$.
For this fixed permutation $\pi$, define a new movement-cost function on $V(G)$ by
$\widehat c_j(v):=c_{\pi(j)}(v)$ for $j\in[\ell]$.
Thus, once $\pi$ is fixed, the cost of using a vertex $v$ in position $j$ is determined.

At this point, we obtain an instance of \textsc{Path Discovery} with exactly $\ell$ tokens:
for each $i\in[\ell]$, introduce one abstract token $y_i$, and let its movement cost to a vertex~$v$ be
$\dist_{M'}(y_i,v):=c_i(v)$.
We may now invoke \Cref{thm:two-graph-discovery-fpt-k} on this
compressed instance with~$\ell$ tokens. 

\medskip

We prove correctness of the construction. 
Suppose first that the algorithm accepts for some choice of $\ell$, $\lambda_V$, $\lambda_T$, and
$\pi$. Then the compressed instance has a feasible discovered $s$-$t$-path
$P=v_1,\dots,v_\ell$
in $G$. 
By the
definition of the costs $c_i(v)$, for each $j\in[\ell]$ there exists a token $x_j\in S$ of
token-color $\pi(j)$ such that
$\dist_M(x_j,v_j)=\widehat c_j(v_j)$.
Because $\lambda_T$ is injective on token-colors, these tokens are pairwise distinct. Hence, we may
assign them to the vertices $v_1,\dots,v_\ell$, obtaining a valid discovery sequence in the original
instance of total cost at most~$b$.

Conversely, suppose that the original instance is a yes-instance. Let
$P=v_1,\dots,v_\ell$
be a discovered simple $s$-$t$-path of minimum movement cost, and let 
$x_1,\dots,x_\ell\in S$
be the pairwise distinct tokens used to occupy the vertices $v_1,\dots,v_\ell$, respectively.
Since $\ell\le r$, by the defining properties of $\mathcal F_\ell$ and~$\mathcal H_\ell$, there exist
$\lambda_V\in\mathcal F_\ell$ and $\lambda_T\in\mathcal H_\ell$ such that $\lambda_V$ is injective on
$\{v_1,\dots,v_\ell\}$ and $\lambda_T$ is injective on $\{x_1,\dots,x_\ell\}$.
Let $\pi\in S_\ell$ be the unique permutation such that
$\lambda_V(v_j)=\pi(j)$ and $\lambda_T(x_j)=\pi(j)$
for every $j\in[\ell]$.
Then for each $j$, $\widehat c_j(v_j)\le \dist_M(x_j,v_j)$.
Hence, the compressed instance constructed from $\ell$, $\lambda_V$, $\lambda_T$, and $\pi$
admits a solution of cost at most the cost of the original solution, and therefore the algorithm
accepts.

\medskip

For the running time, we count as follows. 
For each $\ell\in\{2,\dots,r\}$, the number of pairs of colorings
$(\lambda_V,\lambda_T)\in \mathcal F_\ell\times \mathcal H_\ell$
is $2^{\Oof(\ell)}\log |V(G)| \cdot 2^{\Oof(\ell)}\log |S|
=
2^{\Oof(\ell)}\cdot |V(G)|^{\Oof(1)}$.
For each such pair, we try all $\ell!$ permutations $\pi$, and for each $\pi$ we solve one instance
with exactly $\ell$ tokens using \Cref{thm:two-graph-discovery-fpt-k}.
Hence, the total running time is
\[
\sum_{\ell=2}^{r}
2^{\Oof(\ell)}\cdot \ell!\cdot f(\ell)\cdot |V(G)|^{\Oof(1)}
=
g(r)\cdot |V(G)|^{\Oof(1)}.
\]
Thus, \textsc{Path Discovery} is fixed-parameter tractable parameterized by $r$.
This proves the theorem.
\end{proof}}

We obtain the following two corollaries. 

\begin{corollary}\label[corollary]{crl:param-induced-path}
    \BothPathDiscovery is fixed-parameter tractable when parameterized by  vertex integrity or treedepth of the underlying undirected graph. 
\end{corollary}

\begin{corollary}\label{obs:spd-fpt-diameter}
\ShortestPathDiscovery when all edge weights of $G$ are positive is fixed-parameter tractable when parameterized by the distance from $s$ to $t$ in $G$.
\end{corollary}

\begin{proof}
Let $(G,M,s,t,S,b)$ be an instance of \textsc{Shortest Path Discovery}, and let
$\delta:=\dist_G(s,t)$.
By \Cref{obs:SPD=PD}, it suffices to solve \textsc{Path Discovery} on the shortest-path
subgraph $G^\star$ of $G$.
Since all edge weights of $G$ are positive integers, every shortest $s$-$t$-path in~$G$ has at most
$\delta+1$
vertices.
Hence, every directed $s$-$t$-path in $G^\star$ has at most $\delta+1$ vertices.
Applying \Cref{thm:param-induced-path} to $G^\star$ therefore yields an
$f(\delta)\cdot |V(G)|^{\Oof(1)}$-time algorithm.
\end{proof}

\subsection{Parameter feedback edge set}
\label{sec:fes}

We now consider the parameter feedback edge set. 
Here, by feedback edge set we mean a feedback edge set of the underlying undirected graph~$U$ of $G$. 
We begin with a standard auxiliary lemma: if the target configuration is fixed, then
its reachability under a budget can be decided in polynomial time.
This subroutine was used already in tractability proofs for solution discovery; see, 
the proof of Theorem~3.3 in~\cite{fellows2023solution}. 
We adapt the lemma for the two-graph setting. 
The proof follows by solving the described matching problem.  

\begin{lemma}\label[lemma]{lem:fixed-target-poly}
Let $G$ and $M$ be graphs on the same vertex set $V$, let $S,T\subseteq V$ (as multisets) be two token
configurations with $|T|\le |S|$, and let $b\in\mathbb N$.
Assume that lifting unused tokens is allowed.
Then one can decide in polynomial time whether the target configuration $T$ can be realized
from $S$ within total movement budget at most $b$.

More precisely, let $H$ be the complete bipartite graph with bipartition $(S,T)$, and assign
to every edge $uv\in S\times T$ the weight
$w(u,v):=\dist_M(u,v)$.
Then $T$ can be realized from $S$ within budget at most $b$ if and only if $H$ admits a
matching that saturates $T$ and has total weight at most $b$. 
This problem can be solved in polynomial time. 
\end{lemma}

Next we show that graphs with small feedback edge set contain only few simple $s$-$t$-paths.

\begin{lemma}\label[lemma]{lem:few-st-paths-fes}
Let $G$ be a graph, let $s,t\in V(G)$, and let $F\subseteq E(U)$ be a feedback
edge set of the underlying undirected graph $U$ of $G$ of size $f:=|F|$.
Then the number of simple $s$-$t$-paths in~$G$ is at most
\[
N(f)\ :=\ \sum_{r=0}^{f} \binom{f}{r}\,r!\,2^r \ \le\ (2f+1)^f .
\]
Moreover, every simple $s$-$t$-path is uniquely determined by the ordered list of feedback
edges it uses together with the directions in which it traverses them.
\end{lemma}

\shortlongversion{}{
\begin{proof}
We may assume without loss of generality that every edge of $G$ lies on an $s$-$t$-path (as all other edges may be deleted). 
Let $U$ be the underlying undirected graph of $G$ and let $T:=U-F$ as an undirected graph. 
By our assumption, $T$ is an undirected tree.
From now on, slightly misusing notation, we consider $F$ as the set $\{(u,v),(v,u) \mid \{u,v\}\in F\}$. 

For a simple $s$-$t$-path $P$ in $G$, define its \emph{signature} $\sigma(P)$ as follows.
Let $F(P):=E(P)\cap F$ be the set of feedback edges used by $P$; since $P$ is simple, every
edge is used at most once.
List the edges of $F(P)$ in the order in which they occur along $P$, say
$e_1,\ldots,e_r$, and for each $i\in[r]$ record the direction in which $P$ traverses $e_i$,
say this direction is $x_i y_i$.
Thus,
\[
\sigma(P)=\bigl((e_1=x_1y_1),\ldots,(e_r=x_ry_r)\bigr).
\]

We claim that if $P$ and $P'$ are simple $s$-$t$-paths with $\sigma(P)=\sigma(P')$, then
$P=P'$.

Indeed, write the common signature as $(e_1,\ldots,e_r)$ with orientations
$e_i=x_i y_i$.
Set $p_0:=s$, $q_{r+1}:=t$, and for $i\in[r]$ let $q_i:=x_i$ and $p_i:=y_i$.
Then $P$ consists of:
\begin{itemize}
    \item the unique $T$-path from $p_0=s$ to $q_1=x_1$,
    \item the edge $e_1$,
    \item the unique $T$-path from $p_1=y_1$ to $q_2=x_2$,
    \item the edge $e_2$,
    \item $\dots$,
    \item the unique $T$-path from $p_r=y_r$ to $q_{r+1}=t$.
\end{itemize}

Between two consecutive feedback edges used by $P$,
the path uses no edge of $F$, and hence it lies entirely in $T$. Since $T$ is a tree, the
corresponding simple path is unique. The same argument applies to $P'$, so $P$ and $P'$
coincide on every $T$-segment and on every feedback edge. Hence, $P=P'$, proving the claim.

Thus, the map $P\mapsto \sigma(P)$ is injective, and the number of simple $s$-$t$-paths is at
most the number of possible signatures.

To count signatures, choose $r\in\{0,\ldots,f\}$, choose the $r$ feedback edges used by the path
(at most $\binom{f}{r}$ choices), order them (at most $r!$ choices), and orient each of them
(exactly $2^r$ choices). Therefore, the number of signatures is at most
\[
\sum_{r=0}^{f} \binom{f}{r}\,r!\,2^r
\le
\sum_{r=0}^{f} (2f)^r
\le
(2f+1)^f,
\]
which proves the stated bound.
\end{proof}}

As a corollary, graphs of small feedback edge set admit fixed-parameter algorithms
for both \textsc{Path Discovery} and \textsc{Shortest Path Discovery}, as shown next.

\begin{theorem}\label{thm:path-fes}
\BothPathDiscovery{} is fixed-parameter tractable parameterized by
the feedback edge set number of the underlying undirected graph of the problem graph $G$.
\end{theorem}

\shortlongversion{}{
\begin{proof}
Let $(G,M,s,t,S,b)$ be an instance of \textsc{Path Discovery}. 
\textsc{Shortest Path Discovery} follows by \Cref{obs:SPD=PD}. 
We may assume that every edge of $G$ lies on an $s$-$t$-path (as all other edges may be deleted) and that there exists an $s$-$t$-path, as otherwise we may immedately return false. 
Let $k:=|S|$ and let $d:=\dist_G(s,t)$.
Let $F\subseteq E(U)$ be a minimum feedback edge set of size $f$, where $U$ is the underlying undirected graph of $G$, which can be computed in time $\Oof(n+m)$. 
Let $T:=U-F$, which is an undirected tree.

By \Cref{lem:few-st-paths-fes}, the number of simple $s$-$t$-paths in $G$ is at most $N(f)$,
and each such path is uniquely determined by its signature.
Hence, we can enumerate all candidate signatures
$\sigma=\bigl((e_1=x_1y_1),\ldots,(e_r=x_ry_r)\bigr)$
with pairwise distinct $e_i\in F$, and for each signature reconstruct the corresponding
$s$-$t$-walk by concatenating the unique $T$-paths between the prescribed endpoints and
the prescribed feedback edges.
We keep exactly those signatures that yield a simple $s$-$t$-path $P$ with
$|V(P)|\le k$ (and with weight $d$, respectively).
These are precisely the candidate target (shortest) paths that can potentially be realized from $S$.

For each such path $P$, let $T_P:=V(P)$.
Since the target configuration is now fixed and satisfies $|T_P|\le |S|$, we can decide in
polynomial time, using \Cref{lem:fixed-target-poly}, whether $T_P$ can be realized from $S$
within total movement budget at most $b$ in the two-graph model.
If for some candidate path $P$ this check succeeds, we answer \textsc{YES}; otherwise we
answer \textsc{NO}.

The number of candidate signatures depends only on $f$, and each corresponding check runs in
polynomial time. Therefore, the overall running time is $g(f)\cdot n^{O(1)}$ for some computable
function $g$, and the problem is fixed-parameter tractable parameterized by $f$.
\end{proof}}

\subsection{Parameter treewidth}

Finally, we show that \textsc{Path Discovery} and \textsc{Shortest Path Discovery} are XP when parameterized by treewidth. 
Here, for an instance $(G,M,w_G,w_M,s,t,S,b)$, we measure the treewidth in the graph $U:=\bigl(V(G),E^{\mathrm{und}}(G)\cup E^{\mathrm{und}}(M)\bigr)$, which is the undirected union of $G$ and $M$. 
Our proof is a variation of the proof for \textsc{Vertex Cover Discovery} parameterized by treewidth~\cite{GroblerEtAlICALP24}. 
\shortlongversion{}{
We give a short proof overview and refer to~\cite{GroblerEtAlICALP24} for undefined notation. }

\begin{theorem}\label{thm:path-treewidth-xp}
    \BothPathDiscovery{} can be solved in time $n^{\Oof(tw)}$, hence, it is in~\XP{} when parameterized by $t$.
\end{theorem}

\shortlongversion{
    \begin{proof}
        [Proof Sketch.]
        The algorithm performs a bottom-up dynamic program over a nice tree decomposition of the undirected union of the problem graph and the movement graph.  
        At each bag, the state records, first, the interface of the partial directed \(s\)-\(t\)-path with the separator, including local path roles and the pairing of open path endpoints, and second, a balance function describing how many tokens must still enter or may leave the processed part through each separator vertex.  
        Since both the path interface and the token balance have only \(n^{\Oof(t)}\) possibilities per bag and all introduce, forget, and join transitions are local, the dynamic program runs in time \(n^{\Oof(t)}\), yielding an XP algorithm.
    \end{proof}
}{
\begin{proof}[Proof sketch.]
We consider the case of \textsc{Path Discovery}, as \textsc{Shortest Path Discovery} again follows by considering $G^\star$. 
Let $k:=|S|$. 
We compute a nice tree decomposition $(T,(X_z)_{z\in V(T)})$ of the undirected graph $U$ of width~$\Oof(t)$ in time $2^{\Oof(t)}\cdot n^2$.
Fix a root $r\in V(T)$. For a node~$z$, let $T_z$ be the subtree rooted at~$z$, and
let $V_z:=\bigcup_{x\in V(T_z)} X_x$.
We write $G_z:=G[V_z]$ and $M_z:=M[V_z]$.

The dynamic program processes the decomposition bottom-up. For every bag $X_z$,
we store how a partial directed $s$-$t$-path inside $G_z$ meets the separator
$X_z$, and how tokens are routed inside $M_z$ towards the vertices of this partial
path.

\medskip
\noindent
\textbf{Path interface.}
A partial solution inside $V_z$ induces a collection of pairwise vertex-disjoint directed path segments in $G_z$ whose vertices outside the bag lie in $V_z\setminus X_z$, and whose endpoints lie in $X_z \cup \{s, t\}$.
Since the final object is a single directed $s$-$t$-path, such a partial solution is fully described on the separator by the following information.

For every $v\in X_z$, we record how the partial path inside $V_z$ is incident with $v$: namely, whether no path edge of the directed path is incident with $v$, exactly one outgoing path edge of the partial path is incident with $v$, exactly one incoming path edge of the partial path is incident with $v$, both one incoming and one outgoing path edge of the partial path are incident with $v$, or $v$ is already declared to lie on the directed path solution, although no incident path edge has yet appeared inside $G_z$.
In addition, we store a directed pairing on those vertices of $X_z$ that are endpoints of partial directed path segments, representing how these endpoints are connected by the partial directed path inside $G_z$.
Since $|X_z|=\Oof(t)$, the number of such path-interface types is bounded by a function of $t$, in fact by $t^{\Oof(t)}$.

\medskip
\noindent
\textbf{Token interface.}
Let $P_z$ be the set of path vertices contained in $V_z\setminus X_z$, and let $Q_z:=V(P)\cap X_z$ be the path vertices inside the bag $X_z$.
To realize the partial solution, tokens from $S\cap V_z$ may be moved inside $M_z$ to vertices of $P_z$ and possibly to some vertices of $Q_z$. 
Since every interaction between $V_z$ and other vertices passes through $X_z$, from the viewpoint of the unprocessed part it is enough to know, for every separator vertex $v\in X_z$, the net number of tokens that leave the processed part through $v$ or that still have to enter through~$v$. We therefore store a balance function
\[
\beta:X_z\to\{-k,-k+1,\dots,k\}.
\]
Here $\beta(v)>0$ means that after occupying all directed path vertices in $V_z$, the processed part provides $\beta(v)$ surplus tokens at $v$, whereas $\beta(v)<0$ means that the processed part still needs $-\beta(v)$ tokens to enter through $v$. 
The number of possible balance functions is at most $(2k+1)^{|X_z|}\le n^{\Oof(t)}$.

\medskip
\noindent
\textbf{DP table.}
For every node $z$, every path-interface type $\tau$ on $X_z$, and every balance function~$\beta$, we define
\[
\mathrm{DP}_z(\tau,\beta)\in \mathbb{Z}\cup\{\infty\}
\]
to be the minimum total movement cost inside $M_z$ among all partial solutions in $V_z$ that
\begin{itemize}
    \item induce the path-interface type $\tau$ on $X_z$, and
    \item induce the token balance $\beta$ on $X_z$.
\end{itemize}
If no such partial solution exists, we set $\mathrm{DP}_z(\tau,\beta):=\infty$.

\medskip
\noindent
\textbf{Correctness of the state description.}
The path-interface type $\tau$ is sufficient to summarize the \emph{path structure} of a solution at the separator $X_z$. 
Indeed, $X_z$ separates $V_z\setminus X_z$ from the rest of the graph, and therefore any continuation of the partial path solution outside $V_z$ can interact with the already processed part only through vertices of $X_z$. 
Hence, for the path part of the dynamic program, it suffices to record which vertices of $X_z$ lie on the directed path solution, what their local roles are (unused, isolated used, outgoing boundary, incoming boundary, internal) with respect to the partial path, and how the endpoint vertices of the partial directed path segments are connected inside $V_z$.

The balance function $\beta$ is sufficient to summarize the \emph{token-routing part} of a solution at the separator $X_z$.
Indeed, since $X_z$ separates $V_z\setminus X_z$ from the vertices of the rest of the graph and all token slides that use vertices of $V_z\setminus X_z$ are also fully contained in $M_z$, the only information needed by the outside of $X_z$ is how many tokens are available as a surplus or how many are still required at each separator vertex to enter to $V_z$. 
Tokens are indistinguishable, so their precise internal identities are irrelevant once the minimum movement cost has been fixed.

\medskip
\noindent
\textbf{Transitions.}
The table is filled in the standard way over the nice tree decomposition.

At a leaf node, the bag is empty. 
Hence, there is only one possible path-interface type and one possible balance function, namely the empty ones, and the corresponding DP value is computed directly.

At an introduce node, one new vertex $u$ enters the bag. 
To compute a parent state from a child state, we guess how the partial path is incident with $u$ (unused, isolated used, outgoing boundary, incoming boundary, internal).
At the same time, when applicable, we guess which path edges between $u$ and vertices already present in the bag are used. 
These choices may also change the roles of previously present bag vertices adjacent to $u$. 
In addition, we guess how the directed pairing on the endpoint vertices is updated.
We keep exactly those choices that satisfy the local degree constraints of a directed path and yield an update of the directed pairing on the endpoint vertices consistent with the chosen incidences of $u$ and with the child state.
The token-routing information is guessed analogously by accounting for the contribution of $u$ to the token supply and demand inside the processed part, and the corresponding cost is updated accordingly. 
Since every edge of $G$ or $M$ with one endpoint~$u$ and the other endpoint already processed has its other endpoint in the current bag, all the resulting checks are local to the bag.
For each fixed resulting state, we store only the minimum cost of a partial solution realizing this state. 
This is safe because the unprocessed part interacts with the processed part only through the information encoded by the state; hence two partial solutions with the same state are equivalent for all future transitions, and only the cheaper one needs to be retained.

At a forget node, one vertex $u$ leaves the bag. 
Since $u$ will no longer belong to the separator, the unprocessed part can no longer interact with the processed part through $u$. 
Hence, we keep exactly those child states in which the role of $u$ is already fully resolved: if $u$ lies on the partial path, then all path incidences required at $u$ have already been realized inside the processed part; similarly, the token-routing requirements at $u$ must already be settled inside the processed part, that is, the balance at $u$ must be zero before $u$ is forgotten. 
We then project the child state to the smaller bag by deleting the information associated with $u$.

At a join node, the two children have the same bag $X_z$ and correspond to two disjoint processed subgraphs outside the separator $X_z$. 
For each possible path-interface type of a join node, we therefore consider all pairs of child states whose path information is jointly feasible on $X_z$, that is, such that for every $v\in X_z$, the path incidences contributed by the two children combine to a valid local path type consistent with the parent state and do not violate the local degree constraints of a directed path.
The corresponding parent balance is obtained by pointwise addition of the two child balances, and the corresponding cost is the sum of the two child costs. 
For each resulting state $(\tau,\beta)$, we retain only the minimum cost over all compatible pairs of child states producing it.

At a join node, we consider all pairs of child states for each parent path-interface type $\tau$, of which there are $t^{\Oof(t)}$. 
Since each child bag has at most $n^{\Oof(t)}$ states, the number of pairs of child states is at most $n^{\Oof(t)}$.
For each pair,  the compatibility checks and balance update are carried out in time $n^{\Oof(t)}$. 
Hence, the running time at a join node is $n^{\Oof(t)}$.

\medskip
\noindent
\textbf{Acceptance.}
At the root $r$, there is no outside part.
Hence we accept if there exists a state $(\tau,\beta)$ such that $\tau$ describes one complete directed $s$-$t$-path and no other path segment, $\beta\equiv 0$, and $\mathrm{DP}_r(\tau,\beta)\le b$.
By construction, such a state exists if and only if there is a directed $s$-$t$-path in $G$ that can be realized from $S$ with movement cost at most $b$ in $M$.

\medskip
\noindent
\textbf{Running time.}
Computing the tree decomposition takes time $2^{\Oof(t)}\cdot n^2$. 
For a bag of size $w=\Oof(t)$, the number of path-interface types is $w^{\Oof(w)}$, and the number of balance functions is $(2k+1)^w$. 
Hence, the number of states per bag is bounded by $w^{\Oof(w)}\cdot (2k+1)^w = n^{\Oof(t)}$.
A nice tree decomposition has $\Oof(n)$ bags, and each transition can be processed in $n^{\Oof(t)}$.
Therefore, the overall running time is $n^{\Oof(t)}$, proving the theorem.

The \textsc{Shortest Path Discovery} case follows by restricting the problem graph to the shortest-path subgraph $G^\star$ of $G$, as before. 
\end{proof}}

\section{Hardness results}
\label{sec:hardness}

In this section, we collect our hardness results. We first give a direct NP-hardness
reduction that simultaneously works for \textsc{Path Discovery} and \textsc{Shortest Path Discovery} even in the classical undirected unweighted model, and then explain how the
same construction yields hardness on several restricted graph classes and for several
structural parameters. After that, we incorporate the weighted two-graph hardness
constructions for the diameter and budget parameters.

\subsection{Source problem: \textnormal{\textsc{Circulating Orientation}}}

We reduce from \textsc{Circulating Orientation}, which is defined as follows.

\begin{definition}[\textnormal{\textsc{Circulating Orientation}}]
Let $G=(V,E)$ be an undirected multigraph and let
$w \colon E \to \mathbb{N}$ be an edge-weight function.
For an orientation $D$ of $G$ let $A(D)$ denote the corresponding set of directed arcs. For a vertex $v\in V$, define the
\emph{weighted indegree} of $v$ by
\[
\indeg_w^D(v)\coloneqq\sum_{(u,v)\in A(D)} w(\{u,v\})
\]
and similarly the \emph{weighted outdegree} of $v$ by
\[
    \out_w^D(v)\coloneqq\sum_{(v,u)\in A(D)} w(\{v,u\}).
\]

A \emph{circulating orientation} of $(G,w)$ is an orientation $D$ of $G$
such that
\[
\out_w^D(v)=\indeg_w^D(v)
\qquad\text{for every }v\in V.
\]

Equivalently, $D$ is a circulating orientation if and only if
\[
\out_w^D(v)=\frac{1}{2}\sum_{e \ni v} w(e)
\qquad \text{for every } v \in V.
\]

The problem \textnormal{\textsc{Circulating Orientation}} asks whether a given weighted multigraph~$(G,w)$ admits a circulating orientation.
\end{definition}

The problem is also known as \textsc{Flow Orientation}. 
It is NP-hard on planar graphs~\cite{DidimoLiottaPatrignani2019HVPlanarity}. 
Moreover, the reduction of Jansen et al.~\cite{Jensen2023Planarity} shows W[1]-hardness parameterized by pathwidth even for planar instances.

\subsection{NP-hardness on planar graphs even in the undirected one-graph model}

We now give an NP-hardness reduction for \textsc{(Shortest) Path Discovery} even in the classical undirected unweighted
token-sliding model from \textsc{Circulating Orientation} on planar graphs.
Thus, throughout this subsection, the problem graph and the movement graph coincide, the edge relation is symmetric, and all edges have unit cost. 
For simplicity we treat the graphs as undirected graphs. 

Intuitively, we will create a path that crosses all edges of our original graph twice.
We introduce gadgets that force the choice of subpaths in a way that is equivalent to choosing an edge orientation for the original problem.
All this is done while preserving planarity.

We need the following auxiliary lemma.

\begin{lemma}\label[lemma]{lem:planar-spine}
Let $G$ be a connected plane graph.
Then one can compute in polynomial time a Jordan curve $\gamma$ such that
\begin{enumerate}
    \item for every edge $e\in E(G)$, the curve $\gamma$ intersects the relative
    interior of $e$ in exactly two points;
    \item $\gamma$ is disjoint from $V(G)$; and
    \item $\gamma$ has no other intersection with the drawing of $G$.
\end{enumerate}
Consequently, after subdividing every edge of $G$ at the two intersection points
with $\gamma$, adding edges between consecutive subdivision vertices along
$\gamma$ preserves planarity and makes the new subdivision vertices induce a cycle.
Deleting one edge of this cycle yields a path $P$ that crosses every edge exactly twice.
\end{lemma}

The construction is illustrated in \Cref{fig:skeleton}.

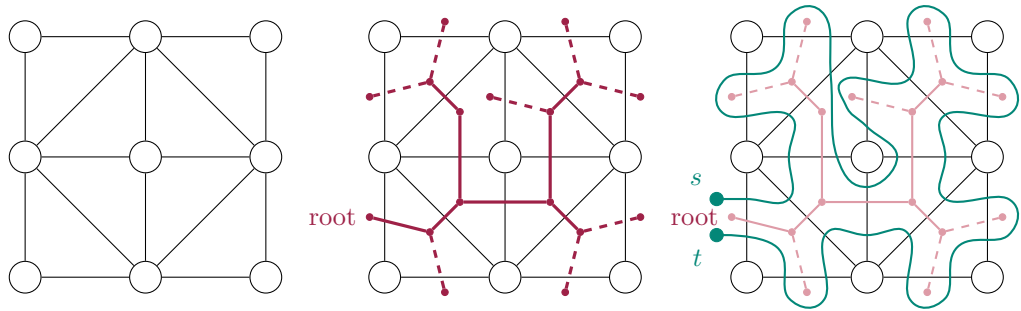
\begin{figure}[h!]
    \centering
    \begin{tikzpicture}[baseline=(v11),mainvertex/.style={draw, circle, minimum height=1.2em, inner sep=0em},
    pathvertex/.style={draw=darkgreen, circle, minimum width=0.5em, inner sep=0em,fill=darkgreen},
    treevertex/.style={draw=darkred, circle, minimum width=0.25em, inner sep=0em,fill=darkred},
    token/.style={draw,rectangle, minimum height=0.25em,  minimum width=0.25em, fill=black, inner sep=0em},longpath/.style={dashed},
    edge/.style={draw},
    treeedge/.style={draw,thick,darkred},
    pathedge/.style={draw,thick,darkgreen},
    xscale=0.4, yscale=0.4]

    \node[mainvertex] (v11) at (0,0) {};

    \node[mainvertex] (v12) at (4,0) {};

    \node[mainvertex] (v13) at (8,0) {};

    \node[mainvertex] (v21) at (0,4) {};

    \node[mainvertex] (v22) at (4,4) {};

    \node[mainvertex] (v23) at (8,4) {};

    \node[mainvertex] (v31) at (0,8) {};

    \node[mainvertex] (v32) at (4,8) {};

    \node[mainvertex] (v33) at (8,8) {};

    \draw[edge] (v11) -- (v12);
    \draw[edge] (v12) -- (v13);

    \draw[edge] (v21) -- (v22);
    \draw[edge] (v22) -- (v23);

    \draw[edge] (v31) -- (v32);
    \draw[edge] (v32) -- (v33);

    \draw[edge] (v11) -- (v21);
    \draw[edge] (v21) -- (v31);

    \draw[edge] (v12) -- (v22);
    \draw[edge] (v22) -- (v32);

    \draw[edge] (v13) -- (v23);
    \draw[edge] (v23) -- (v33);

    \draw[edge] (v21) -- (v12);
    \draw[edge] (v23) -- (v32);
    \draw[edge] (v21) -- (v32);
    \draw[edge] (v23) -- (v12);

\end{tikzpicture}
    \begin{tikzpicture}[baseline=(v11),mainvertex/.style={draw, circle, minimum height=1.2em, inner sep=0em},
    pathvertex/.style={draw=darkgreen, circle, minimum width=0.5em, inner sep=0em,fill=darkgreen},
    treevertex/.style={draw=darkred, circle, minimum width=0.25em, inner sep=0em,fill=darkred},
    token/.style={draw,rectangle, minimum height=0.25em,  minimum width=0.25em, fill=black, inner sep=0em},longpath/.style={dashed},
    edge/.style={draw},
    treeedge/.style={draw,very thick,darkred},
    pathedge/.style={draw,very thick,darkgreen},
    finger/.style={dashed},
    xscale=0.4, yscale=0.4]

    \node[mainvertex] (v11) at (0,0) {};

    \node[mainvertex] (v12) at (4,0) {};

    \node[mainvertex] (v13) at (8,0) {};

    \node[mainvertex] (v21) at (0,4) {};

    \node[mainvertex] (v22) at (4,4) {};

    \node[mainvertex] (v23) at (8,4) {};

    \node[mainvertex] (v31) at (0,8) {};

    \node[mainvertex] (v32) at (4,8) {};

    \node[mainvertex] (v33) at (8,8) {};

    \draw[edge] (v11) -- (v12);
    \draw[edge] (v12) -- (v13);

    \draw[edge] (v21) -- (v22);
    \draw[edge] (v22) -- (v23);

    \draw[edge] (v31) -- (v32);
    \draw[edge] (v32) -- (v33);

    \draw[edge] (v11) -- (v21);
    \draw[edge] (v21) -- (v31);

    \draw[edge] (v12) -- (v22);
    \draw[edge] (v22) -- (v32);

    \draw[edge] (v13) -- (v23);
    \draw[edge] (v23) -- (v33);

    \draw[edge] (v21) -- (v12);
    \draw[edge] (v23) -- (v32);
    \draw[edge] (v21) -- (v32);
    \draw[edge] (v23) -- (v12);

    \node[treevertex] (t1) at (2,-0.5) {};
    \node[treevertex,label=left:\textcolor{darkred}{root}] (t2) at (-0.5,2) {};
    \node[treevertex] (tx) at (1.5,1.5) {};
    \node[treevertex] (t3) at (2.5,2.5) {};
    \node[treevertex] (t4) at (2.5,5.5) {};
    \node[treevertex] (t5) at (5.5,2.5) {};
    \node[treevertex] (t6) at (6.5,1.5) {};
    \node[treevertex] (t7) at (6,-0.5) {};
    \node[treevertex] (t8) at (8.5,2) {};
    \node[treevertex] (t9) at (5.5,5.5) {};
    \node[treevertex] (t10) at (3.5,6) {};
    \node[treevertex] (t11) at (6.5,6.5) {};
    \node[treevertex] (t12) at (8.5,6) {};
    \node[treevertex] (t13) at (6,8.5) {};
    \node[treevertex] (t14) at (1.5,6.5) {};
    \node[treevertex] (t15) at (2,8.5) {};
    \node[treevertex] (t16) at (-0.5,6) {};

    \draw[treeedge] (t2) -- (tx);
    \draw[treeedge,finger] (t1) -- (tx);
    \draw[treeedge] (tx) -- (t3);
    \draw[treeedge] (t3) -- (t4);
    \draw[treeedge] (t3) -- (t5);
    \draw[treeedge] (t6) -- (t5);
    \draw[treeedge,finger] (t6) -- (t7);
    \draw[treeedge,finger] (t6) -- (t8);
    \draw[treeedge] (t5) -- (t9);
    \draw[treeedge,finger] (t10) -- (t9);
    \draw[treeedge] (t11) -- (t9);
    \draw[treeedge,finger] (t11) -- (t12);
    \draw[treeedge,finger] (t11) -- (t13);
    \draw[treeedge] (t14) -- (t4);
    \draw[treeedge,finger] (t14) -- (t15);
    \draw[treeedge,finger] (t14) -- (t16);
\end{tikzpicture}
    \begin{tikzpicture}[baseline=(v11),mainvertex/.style={draw, circle, minimum height=1.2em, inner sep=0em},
    pathvertex/.style={},
    st/.style={draw=darkgreen, circle, minimum width=0.4em, inner sep=0em,fill=darkgreen},
    treevertex/.style={draw=lightred, circle, minimum width=0.25em, inner sep=0em,fill=lightred},
    token/.style={draw,rectangle, minimum height=0.25em,  minimum width=0.25em, fill=black, inner sep=0em},longpath/.style={dashed},
    edge/.style={draw},
    treeedge/.style={draw,thick,lightred},
    pathedge/.style={draw,thick,darkgreen,overlay},
    finger/.style={dashed},
    xscale=0.4, yscale=0.4]
    \newcounter{pathlastvertex}
    \newcommand\startpath{\edef\lastpend{\thepathlastvertex}}

    \newcommand\pathvertex[3][]{
        \stepcounter{pathlastvertex}

        \node[pathvertex,#1] (p\thepathlastvertex) at (#2,#3) {};

    }

    \node[mainvertex] (v11) at (0,0) {};

    \node[mainvertex] (v12) at (4,0) {};

    \node[mainvertex] (v13) at (8,0) {};

    \node[mainvertex] (v21) at (0,4) {};

    \node[mainvertex] (v22) at (4,4) {};

    \node[mainvertex] (v23) at (8,4) {};

    \node[mainvertex] (v31) at (0,8) {};

    \node[mainvertex] (v32) at (4,8) {};

    \node[mainvertex] (v33) at (8,8) {};

    \draw[edge] (v11) -- (v12);
    \draw[edge] (v12) -- (v13);

    \draw[edge] (v21) -- (v22);
    \draw[edge] (v22) -- (v23);

    \draw[edge] (v31) -- (v32);
    \draw[edge] (v32) -- (v33);

    \draw[edge] (v11) -- (v21);
    \draw[edge] (v21) -- (v31);

    \draw[edge] (v12) -- (v22);
    \draw[edge] (v22) -- (v32);

    \draw[edge] (v13) -- (v23);
    \draw[edge] (v23) -- (v33);

    \draw[edge] (v21) -- (v12);
    \draw[edge] (v23) -- (v32);
    \draw[edge] (v21) -- (v32);
    \draw[edge] (v23) -- (v12);

    \node[treevertex] (t1) at (2,-0.5) {};
    \node[treevertex,label=left:\textcolor{darkred}{root}] (t2) at (-0.5,2) {};
    \node[treevertex] (tx) at (1.5,1.5) {};
    \node[treevertex] (t3) at (2.5,2.5) {};
    \node[treevertex] (t4) at (2.5,5.5) {};
    \node[treevertex] (t5) at (5.5,2.5) {};
    \node[treevertex] (t6) at (6.5,1.5) {};
    \node[treevertex] (t7) at (6,-0.5) {};
    \node[treevertex] (t8) at (8.5,2) {};
    \node[treevertex] (t9) at (5.5,5.5) {};
    \node[treevertex] (t10) at (3.5,6) {};
    \node[treevertex] (t11) at (6.5,6.5) {};
    \node[treevertex] (t12) at (8.5,6) {};
    \node[treevertex] (t13) at (6,8.5) {};
    \node[treevertex] (t14) at (1.5,6.5) {};
    \node[treevertex] (t15) at (2,8.5) {};
    \node[treevertex] (t16) at (-0.5,6) {};

    \draw[treeedge] (t2) -- (tx);
    \draw[treeedge,finger] (t1) -- (tx);
    \draw[treeedge] (tx) -- (t3);
    \draw[treeedge] (t3) -- (t4);
    \draw[treeedge] (t3) -- (t5);
    \draw[treeedge] (t6) -- (t5);
    \draw[treeedge,finger] (t6) -- (t7);
    \draw[treeedge,finger] (t6) -- (t8);
    \draw[treeedge] (t5) -- (t9);
    \draw[treeedge,finger] (t10) -- (t9);
    \draw[treeedge] (t11) -- (t9);
    \draw[treeedge,finger] (t11) -- (t12);
    \draw[treeedge,finger] (t11) -- (t13);
    \draw[treeedge] (t14) -- (t4);
    \draw[treeedge,finger] (t14) -- (t15);
    \draw[treeedge,finger] (t14) -- (t16);

    \newcommand\kasp{0.7}
    \newcommand\sepp{0.4}
    \node[st,minimum width=0.5em, label=above left:\textcolor{darkgreen}{$s$}](s) at (-1, 3-\sepp){};
    \node[st,minimum width=0.5em, label=below left:\textcolor{darkgreen}{$t$}](t) at (-1, 1+\sepp){};
    \pathvertex{0}{3-\sepp}

    \pathvertex {1+\sepp}{3-\sepp}
    \pathvertex{1+\sepp}{4}
    \pathvertex{1+\sepp}{5+\sepp}
    \pathvertex{0-\kasp}{5+\sepp}
    \pathvertex{0-\kasp}{7-\sepp}
    \pathvertex{1}{7}
    \pathvertex{1+\sepp}{8+\kasp}
    \pathvertex{3-\sepp}{8+\kasp}
    \pathvertex{{3-(\sepp/2)}}{7-\sepp}
    \pathvertex{3}{4}
    \pathvertex{4}{3}
    \pathvertex{5}{4}
    \pathvertex{4}{5}
    \pathvertex{4-\kasp}{5+\kasp}
    \pathvertex{4-\kasp}{7-\sepp}
    \pathvertex{5+\sepp}{7-\sepp}
    \pathvertex{5+\sepp}{8+\kasp}
    \pathvertex{7-\sepp}{8+\kasp}
    \pathvertex{7}{7}
    \pathvertex{8+\kasp}{7-\sepp}
    \pathvertex{8+\kasp}{5+\sepp}
    \pathvertex{7-\sepp}{5+\sepp}
    \pathvertex{7-\sepp}{4}
    \pathvertex{7-\sepp}{3-\sepp}
    \pathvertex{8+\kasp}{3-\sepp}
    \pathvertex{8+\kasp}{1+\sepp}
    \pathvertex{7}{1}
    \pathvertex{7-\sepp}{0-\kasp}
    \pathvertex{5+\sepp}{0-\kasp}
    \pathvertex{5+\sepp}{1+\sepp}
    \pathvertex{4}{1+\sepp}
    \pathvertex{3-\sepp}{1+\sepp}
    \pathvertex{3-\sepp}{0-\kasp}
    \pathvertex{1+\sepp}{0-\kasp}
    \pathvertex{1}{1}
    \pathvertex{0-\kasp}{1+\sepp}

\draw[smooth, tension=0.8, pathedge]  plot coordinates { (s)  (p1) (p2) (p3) (p4) (p5) (p6) (p7) (p8) (p9) (p10)   (p11) (p12) (p13) (p14) (p15) (p16) (p17) (p18) (p19) (p20)   (p21) (p22) (p23) (p24) (p25) (p26) (p27) (p28) (p29) (p30) (p31) (p32) (p33) (p34) (p35) (p36) (t)};
\end{tikzpicture}
    \caption{The original graph, the spanning tree of its dual graph (red) with detours (dotted) and constructing a path (green) crossing each edge of a planar graph twice while moving around the spanning tree including its detours of its dual graph on a Jordan curve}
    \label{fig:skeleton}
\end{figure}

\shortlongversion{}{
\begin{proof}
Fix the given plane embedding of $G$.
Let $G^\ast$ be the dual graph, with one dual vertex for every face of $G$.
Choose a spanning tree $T^\ast$ of $G^\ast$ after deleting loops.
This can be done in polynomial time.
We draw $T^\ast$ in the usual way: each dual vertex is first placed in the interior of the corresponding face of $G$, and each dual edge is drawn so that it crosses its corresponding primal edge exactly once and is otherwise contained in the interiors of the two incident faces.
Since $T^\ast$ is a subgraph of the planar dual and is a tree, this drawing is planar, in particular without self-crossings.

Let N be a sufficiently small closed regular neighborhood of the drawing of $T^\ast$.
Since $T^\ast$ is a tree, N is homeomorphic to a closed disc.
Consequently, the boundary of $N$, denoted~$\partial N$, is a Jordan curve.
For every primal edge $e$ whose dual edge belongs to
$T^\ast$, the curve~$\partial N$ crosses $e$ exactly twice: once on each side of the
dual arc crossing $e$.

Now consider an edge $e\in E(G)$ whose corresponding dual edge is not in $T^\ast$.
Choose one of the faces incident with $e$, say $F_e$.
Since the dual vertex corresponding to $F_e$ lies in the tree $T^\ast$, the
neighborhood $N$ meets the face $F_e$.
Inside $F_e$, attach to $\partial N$ a very thin detour which runs from
$\partial N$ to the edge $e$, crosses $e$ once into the face on the other side of $e$, turns around in a tiny neighborhood of $e$, crosses $e$ a second time back into $F_e$, and then returns and merges with $\partial N$, see Figure 2. 
This operation modifies the boundary curve only locally (in particular, it is done at most once per edge) and keeps it a Jordan curve.
We perform this operation independently for every edge whose dual edge is not in $T^\ast$.
The detours can be chosen pairwise disjoint because each face contains only finitely many assigned edges, and the detours may be drawn in sufficiently small, mutually disjoint spaces.
Moreover, the detours are drawn inside faces of $G$, except for the two prescribed crossings with the assigned edge.

After all detours have been inserted, we obtain the Jordan curve $\gamma$.
Every edge whose dual edge lies in $T^\ast$ is crossed twice by the original
boundary $\partial N$ and is not affected by the later detours.
Every remaining edge is crossed exactly twice by its own detour and is not crossed by any other part of the construction.
Thus, $\gamma$ crosses every edge of $G$ exactly twice.
All intersections are transverse crossings with edges of $G$, and away from these crossings the curve lies in the faces of $G$ and has no self-intersections.
Hence, the union of $G$ and $\gamma$ is planar apart from the prescribed crossings.

Finally, subdivide every edge of $G$ at its two crossing points with $\gamma$.
Between consecutive subdivision vertices along $\gamma$, add an edge drawn on the
corresponding sub-curve of $\gamma$.
Since~$\gamma$ is simple and meets $G$ only at subdivision vertices, this implies planarity, and the new subdivision vertices induce a cycle.
Deleting one edge of this cycle gives a path through all subdivision vertices.
The construction uses only a spanning tree of the dual and local detours inside faces, so it is computable in polynomial time.
\end{proof}}

\begin{theorem}\label{thm:planar-hardness}
    \BothPathDiscovery{} in the classical undirected unweighted
    token-sliding model on planar graphs is NP-hard.
\end{theorem}

\begin{figure}[t]
    \centering
    \begin{tikzpicture}[baseline={(0,0)},mainvertex/.style={draw, circle, minimum height=1.7em},hiddenmainvertex/.style={minimum height=1.3em}, weight/.style={},
    pathvertex/.style={draw=darkgreen, fill=darkgreen, circle, minimum width=0.25em},token/.style={draw,rectangle, minimum height=0.25em,  minimum width=0.25em, fill=black, inner sep=0em},
    longpath/.style={dashed},
    edge/.style={draw},
    pathedge/.style={draw,very thick,darkgreen,overlay}
]
    \newcommand\yscale{0.5}
    \node[mainvertex,label=left:$v_1$] (v1) at (0,6*\yscale) {};

    \node[mainvertex, label=left:$v_2$] (v2) at (0,-6*\yscale) {};

    \node[pathvertex,label=\color{darkgreen}$a_s$] (qs) at (-2,3*\yscale) {};
    \node[pathvertex,label=\color{darkgreen}$a_t$] (qt) at (2,3*\yscale) {};
    \node[pathvertex,label=below:\color{darkgreen}$c_s$] (rs) at (2,-3*\yscale) {};
    \node[pathvertex,label=below:\color{darkgreen}$c_t$] (rt) at (-2,-3*\yscale) {};

    \draw[edge, bend right] (v1) to node[left] {$1$} (v2);
    \draw[edge] (v1) to node[left] {$2$} (v2);
    \draw[edge, bend left] (v1) to node[left] {$3$} (v2);

    \draw[pathedge] (qs) -- (qt);
    \draw[pathedge] (rs) -- (rt);
    \node[draw=none,label={[label distance=0.25cm]below:\color{darkgreen}$\leftarrow s$}] at (qs) {};
    \node[draw=none,label={[label distance=0.25cm]above:\color{darkgreen}$\leftarrow t$}] at (rt) {};

    \draw[pathedge,dashed, decorate,decoration=snake] (qt) -- (rs);

\end{tikzpicture}\hfill
    \newcommand\token{\tikz{\node[token]{};}}
\begin{tikzpicture}[baseline=(v2),mainvertex/.style={draw, circle, minimum height=1.7em},hiddenmainvertex/.style={minimum height=1.3em}, weight/.style={},
    pathvertex/.style={draw=darkgreen, thick, circle, minimum width=0.25em},token/.style={draw,rectangle, minimum height=0.25em,  minimum width=0.25em, fill=black, inner sep=0em},longpath/.style={dashed},
pathedge/.style={draw,thick,darkgreen,overlay}
]
    \newcommand\scale{0.5}
    \newcommand\mainvertex[5][missing]{
        \node[mainvertex,label=left:#1] (v#2) at (#3*\scale, #4*\scale){};
        \tokens{#3}{#4}{0.46em}{#5}
    }
    \newcommand\tokens[4]{
        \foreach \i in {1, ..., #4} {    \node[token] (n\i) at ($(#1*\scale, #2*\scale)+({(\i - 1) * 360 / #4}:#3)$) {};  }
    }
    \newcounter{xcoordc}
    \newcommand\nextx{\stepcounter{xcoordc}}
    \newcommand\xcoord{{\value{xcoordc}*\scale}}
    \newcommand\ycoord{0}
    \newcommand{\accessnode}[1]{\the\numexpr (#1)/2\relax}
    \newcommand\gadget[5][]{
        \nextx
        \node[pathvertex] (y#2#3#4) at (\xcoord,{\ycoord}) {};
        \draw[pathedge] (\lastvertexname) -- (y#2#3#4);
        \renewcommand\lastvertexname{y#2#3#4}
        \node[token] at (\lastvertexname) {};

        \foreach \i in {1, ..., #5} {
            \nextx
            \node[pathvertex,#1] (w#2#3#4#2\i) at (\xcoord,\ycoord+\scale) {};
            \node[pathvertex,#1] (w#2#3#4#3\i) at (\xcoord,\ycoord-\scale) {};
            \draw[longpath] (v#2) -- (w#2#3#4#2\i);
            \draw[longpath] (v#3) -- (w#2#3#4#3\i);

        }

        \nextx
        \node[pathvertex] (w#2#3#4) at (\xcoord,\ycoord) {};
        \node[token] at (w#2#3#4) {};
        \draw[pathedge] (\lastvertexname) -- (w#2#3#4#21);
        \draw[pathedge] (\lastvertexname) -- (w#2#3#4#31);

        \renewcommand\lastvertexname{w#2#3#4}

        \ifnum
            #5>1
            \foreach \i in {1, ..., \the\numexpr#5-1\relax} {
                \draw[pathedge] (w#2#3#4#2\i) -- (w#2#3#4#2\the\numexpr\i+1\relax);
                \draw[pathedge] (w#2#3#4#3\i) -- (w#2#3#4#3\the\numexpr\i+1\relax);
            }
        \fi

        \draw[pathedge] (w#2#3#4) -- (w#2#3#4#2#5);
        \draw[pathedge] (w#2#3#4) -- (w#2#3#4#3#5);
    }

    \mainvertex[$v_1$]{1}{6}{3}{3};
    \mainvertex[$v_{1,2}$]{2}{6}{-3}{6};
    \mainvertex[$v_{2}$]{3}{6}{-9}{3};

    \newcommand\intermediatevertex[2]{
        \foreach \i in {1, ..., #2} {
            \nextx
            \node[pathvertex] (#1\i) at (\xcoord,0) {};
            \node[token] at (#1\i) {};
            \draw (\lastvertexname) -- (#1\i);

            \xdef\lastvertexname{#1\i}
        }
    }

    \node[label=\color{darkgreen}$a_s$,pathvertex] (stop) at (\xcoord,\ycoord) {};
    \node[token] at (stop) {};
    \newcommand\lastvertexname{stop}
    \gadget{1}{2}{a}{1}
    \gadget{1}{2}{b}{2}
    \gadget{1}{2}{c}{3}

    \nextx

    \node[label=\color{darkgreen}$a_t$,pathvertex] (ttop) at (\xcoord,\ycoord) {};
    \node[token] at (ttop) {};
    \draw[pathedge] (\lastvertexname) -- (ttop);

    \setcounter{xcoordc}{0}

    \renewcommand\ycoord{-6*\scale}

    \node[label=below:\color{darkgreen}$c_t$,pathvertex] (sbottom) at (\xcoord,\ycoord) {};
    \node[token] at (sbottom) {};
    \renewcommand\lastvertexname{sbottom}
    \gadget{2}{3}{a}{1}
    \gadget{2}{3}{b}{2}
    \gadget{2}{3}{c}{3}

    \nextx

    \node[label=below:\color{darkgreen}$c_s$,pathvertex] (tbottom) at (\xcoord,\ycoord) {};
    \node[token] at (tbottom) {};
    \draw[pathedge] (\lastvertexname) -- (tbottom);

    \node[draw=none,label={[label distance=0.5cm]above:\color{darkgreen}$\leftarrow t$}] at (sbottom) {};
    \node[draw=none,label={[label distance=0.5cm]below:\color{darkgreen}$\leftarrow s$}] at (stop) {};
    \draw[pathedge,dashed, decorate,decoration=snake] (ttop) -- (tbottom);

\end{tikzpicture}
    \caption{Example of multi-edges of weights $1$, $2$ and $3$ between two vertices in $H$ and a planar gadget that would replace them in $G$.
    Only the token capacity for these edges is shown in $v_1$ and $v_2$, other adjacent edges would contribute more tokens.}
    \label{fig:crossing-gadget}
\end{figure}
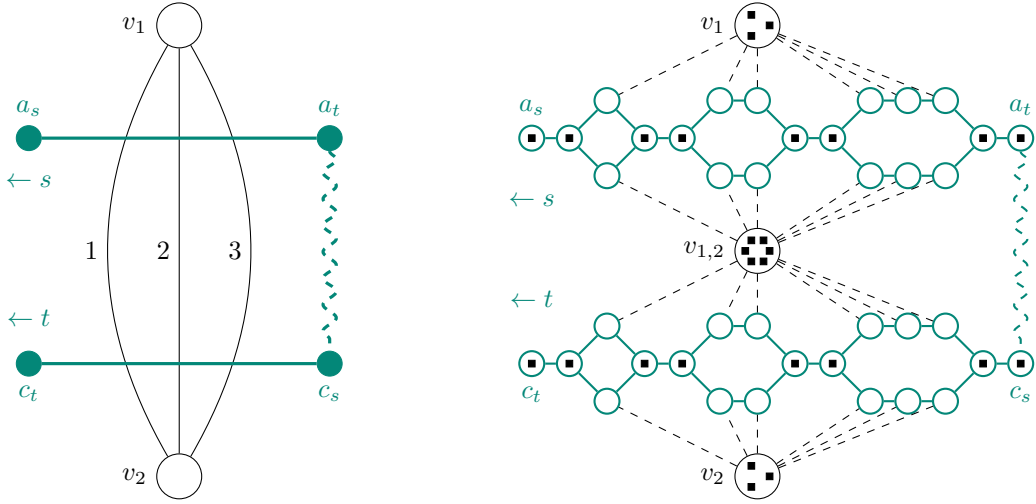

\begin{proof}
We start the proof by describing our construction. Let $(H,w)$ be an instance of \textsc{Circulating Orientation}, where $H$ is a planar graph allowing multi-edges with fixed plane embedding and $w\colon E\to \mathbb{N}$ is an edge-weight function.
For a vertex $v\in V(H)$, let $W(v)\coloneqq\sum_{e\ni v} w(e)$ denote the total of its incident weights.
If $W(v)$ is odd for some vertex~$v$, then the instance is trivially negative.
Hence, we may assume throughout that $W(v)$ is even for every vertex~$v$.
Set $W\coloneqq\sum_{e\in E(H)} w(e)$ and $L\coloneqq2|E(H)|+2W+1$.
We construct an instance $(G,s,t,S,b)$ of \textsc{(Shortest) Path Discovery} as follows.

We incrementally construct the planar graph $G$. It initially contains a copy of $V(H)$.
On every vertex $v \in G$, we place $W(v)/2$ tokens to make it a reservoir vertex containing tokens. Now we add the vertices and paths described below. We apply \Cref{lem:planar-spine} to find the Jordan curve $\gamma$ and the path $P$ crossing every edge exactly twice. We add vertices $s$ and $t$ to the graph and place a token on each of $s$ and $t$. In the following we will describe on how to use the path $P$ to build and connect planar gadgets for each edge $e \in E(H)$.

For every pair of vertices $v_1$, $v_2$ connected by multi-edges $F \subseteq E(H)$, we create a multi-edge gadget.
An example for such gadgets is given in~\Cref{fig:crossing-gadget}.
\begin{itemize}
    \item We create a reservoir vertex $v_{1,2}$ and place tokens equal to the edges capacity $\sum_{e\in F} w(e)$ on it.
    \item For every $e \in F$, we create two pairs of alternative paths, $Q^e_{1}$ and $Q^e_{2}$, from $q^e_s$ to $q^e_t$ and $R^e_{1}$ and $R^e_{2}$ from $r^e_s$ to $r^e_t$. 
    Each of these alternative paths contains $w(e)$ internal vertices. 
    Choosing one of the two will be equivalent to choosing an edge direction.
    We place tokens on $q^e_s,q^e_t,r^e_s,r^e_t$.
    \item We now connect all internal vertices of $Q^e_1$ with $v_1$, all internal vertices of $R^e_2$ with $v_2$ and all internal vertices of $Q^e_2$ and $R^e_1$ with $v_{1,2}$ via long subdivided paths of length $L$. 
    This prevents shortcuts on our intended $s$-$t$ paths while still being quadratic in size.
    \item Finally, we connect these gadgets as given by the path $P$: If $P$ crosses some edge $e$ for the first (second) time just before crossing $e'$ for the first (second) time, we identify $q^e_t$~($r^e_t$) with $q^{e'}_s$ ($r^{e'}_s$).
    \end{itemize}

Finally, we now complete our $s$-$t$ path by connecting $s$ to $q^e_s$ of the first edge $e$ crossed by~$P$ and $r^{e'}_t$ to $t$ for the last edge $e'$ touched by $P$.
By our construction, the shortest paths from $s$ to $t$ exactly have length $L$, and never use any of the reservoir vertices $v_1$, $v_2$, and $v_{1,2}$. This construction preserves planarity, see~\Cref{fig:crossing-gadget}.

Let $S$ be the resulting multiset of token positions, and set
$b\coloneqq2W\cdot L$.

In the following, we show correctness of our construction.
For any shortest $s$-$t$ path of length $L$, there are exactly $2W$ missing tokens.
These exactly have to come from the reservoir vertices.
As we have introduced subdivided paths of length $L$ and set our budget accordingly, the tokens can only move from their reservoir vertices to their corresponding internal vertices of the alternative gadget paths.
Moving a single token further than that would prevent the path from being constructed as we would exceed the budget elsewhere and all tokens are needed to form the path.

We show that the constructed instance $(G,s,t,S,b)$ is positive if and only if
$(H,w)$ admits a circulating orientation.

First, note that every $s$-$t$ path in $G$ must traverse the $Q$-gadgets and $R$-gadgets corres\-ponding to edges in order.
Inside each edge gadget, such a path must choose exactly one of the two internally disjoint branches.
Initially, all vertices $q_s^e, q_t^e, r_s^e, r_t^e$ for every $e \in E(H)$ as well as~$s$ and $t$ carry tokens.
Thus, any feasible discovered path requires exactly $2W$ additional occupied vertices.
The only tokens not already placed on mandatory path vertices are the reservoir tokens.


Each reservoir token sent to an internal gadget vertex must traverse one of the
attachment paths, each of length exactly $L$.
Therefore, moving one such token to its destination costs exactly $L$.
Since a feasible discovered path requires exactly $2W$ such internal vertices to be
occupied, every feasible solution of cost at most $b=2W\cdot L$
must move exactly $2W$ tokens, and each of them must use a shortest attachment path.
In particular, no token can afford any detour, and no token can first move through
one gadget and then continue elsewhere.

Now fix an edge $e=v_1v_2 \in E(H)$.
If the final $s$-$t$ path uses the branch $Q^e_{1}$, then all~$w(e)$ internal vertices of that branch must be occupied.
By construction, every such token must come from the reservoir $v_1$.
Hence, choosing $Q^e_{1}$ consumes exactly $w(e)$ tokens from~$v_1$.
Similarly, choosing $Q^e_2$ consumes exactly $w(e)$ tokens from $v_{1,2}$. For the gadgets $R^e$ the situation is analogous.
Consequently, a feasible solution determines for every vertex $v_1$ a set of edges $e$ for which the branches $Q^e_1$ are used. We orient these edges $e$ to $v_1$. Finally, the edges for which the shortest path uses branch $Q^e_2$ are directed to $v_2$. Let $k$ be the number of tokens used from $v_1$ on the branches $Q^e_1$. Clearly, we use an additional $\sum_{e \in F} w(e) -k$ tokens from $v_{1,2}$. Consequently, we have to use $\sum_{e \in F} w(e) -k$ tokens from $v_2$ for the $R^e_2$ branches.

As each vertex $v$ initially contains exactly $W(v)/2$ tokens and every feasible
solution must use all reservoir tokens, we obtain
\[
\sum_{e\text{ oriented out of }v} w(e)=W(v)/2.
\]
Thus, the chosen orientations form a circulating orientation of $(H,w)$.

Conversely, suppose that $(H,w)$ admits a circulating orientation.
For every edge $e=v_1v_2$, if $e$ is oriented from $v_2$ to $v_1$, let
the final $s$-$t$ path use the branch $Q^e_1$ and the branch $R^e_1$; otherwise let it use the
branches $Q^e_2$ and $R^e_2$.
This determines a feasible $s$-$t$ path through the gadget chain and the budget to move the tokens is exactly used.

Therefore, the constructed instance is positive if and only if $(H,w)$ admits a
circulating orientation.
%
\end{proof}

\subsection{\XNLP-hardness for pathwidth and cutwidth even in the undirected one-graph model}
\label{subsec:xnlp-hardness-pathwidth-cutwidth}

We prove that \textsc{(Shortest) Path Discovery} is \XNLP-hard already in the
classical undirected unweighted token-sliding model when parameterized by
pathwidth, and even when parameter\-ized by cutwidth.

We recall the two width parameters that we use. A \emph{path decomposition} of a
graph $G$ is a sequence $\mathcal{B}=(B_1,\ldots,B_r)$ of vertex sets, called
bags, such that every vertex of $G$ appears in some bag, every edge of $G$ has
both endpoints in some bag, and for every vertex $v\in V(G)$, the set of indices
$i$ with $v\in B_i$ is an interval. The width of $\mathcal{B}$ is
$\max_i |B_i|-1$, and the \emph{pathwidth} of $G$, denoted $\pw(G)$,
is the minimum width of a path decomposition of $G$.

A \emph{linear layout} of a graph $G$ is an ordering
$\pi=(v_1,\ldots,v_n)$ of $V(G)$. The width of $\pi$ is
\[
   \max_{i\in[n-1]}
   |\{uv\in E(G) : \pi^{-1}(u)\le i < \pi^{-1}(v)
       \text{ or } \pi^{-1}(v)\le i < \pi^{-1}(u)\}|.
\]
The \emph{cutwidth} of $G$, denoted $\ctw(G)$, is the minimum width
of a linear layout of $G$.

We use the standard relation between pathwidth and cutwidth. Let
$\Delta(G)$ denote the maximum degree of $G$. Since pathwidth equals vertex
separation~\cite{Kinnersley1992}, one obtains
\[
   \pw(G) \le \ctw(G)
   \qquad\text{and}\qquad
   \ctw(G) \le \Delta(G)\cdot(\pw(G)+1).
\]
\shortlongversion{}{
Indeed, every linear layout of cutwidth $c$ has vertex separation at most $c$,
because every vertex that has already appeared and still has a neighbour to the
right contributes at least one edge crossing the cut. Conversely, from a path
decomposition of width $p$, order vertices by their last bag. At every cut, all
edges crossing the cut have their left endpoint in the current bag, and hence
there are at most $\Delta(G)(p+1)$ such edges. Thus, on graph classes of bounded
maximum degree, bounded pathwidth and bounded cutwidth are functionally
equivalent.}

The reduction is from \textsc{Circulating Orientation} parameterized by
pathwidth, which is \XNLP-complete~\cite{BodlaenderSzilagyi2025XALPPlanar}, and similar to the reduction for planar graphs. 

\shortlongversion{\stepcounter{theorem}}{
We first need a simple way to obtain a bounded-pathwidth spine through the edges
of the source graph.

\begin{lemma}\label[lemma]{lem:pathwidth-spine-closure}
Let $H$ be a graph with $\pw(H)\le k$. Given a path decomposition
of $H$ of width at most $k$, one can construct in polynomial time a graph
$H^\star$ with $\pw(H^\star)\le k+2$ such that 
every edge of $H$ is subdivided exactly once, and the subdivision vertices are
connected into one path. 
\end{lemma}

\begin{proof}
Let $\mathcal{B}=(B_1,\ldots,B_r)$ be a path decomposition of $H$ of width at
most $k$. For every edge $e=uv\in E(H)$, let $p(e)$ be an index such that
$u,v\in B_{p(e)}$. Order the edges as
$e_1,\ldots,e_m$
so that
$p(e_1)\le p(e_2)\le \cdots \le p(e_m)$.
For every $j\in[m]$, subdivide $e_j$ once by a new vertex $x_j$. Then add the
edges $x_jx_{j+1}$ for $j\in[m-1]$.
Thus, the subdivision vertices induce the path
$x_1,x_2,\ldots,x_m$.

We construct a path decomposition of the resulting graph $H^\star$. We keep the
original bags, but around the position $p(e_j)$, we additionally introduce the
subdivision vertex $x_j$, and between $p(e_j)$ and $p(e_{j+1})$, we carry
$x_j$ forward. More explicitly, for every $j<m$, the vertex~$x_j$ is added to
all bags from position $p(e_j)$ up to position $p(e_{j+1})$, and $x_m$ is added
to the bag at position $p(e_m)$. If several edges have the same value of $p(e)$,
we insert the corresponding bags consecutively at that position.

Every original vertex appears in the same interval as before. Every subdivision
vertex~$x_j$ appears in an interval. The edges $u_jx_j$ and $x_jv_j$ are covered
at a bag containing both endpoints of $e_j=u_jv_j$, and the edge $x_jx_{j+1}$ is
covered at the position where $x_{j+1}$ is introduced, since there we keep
$x_j$ and introduce $x_{j+1}$. Every bag contains one original bag and at most
two subdivision vertices. Hence, its size is at most $k+3$, and the width is at
most $k+2$.
\end{proof}}

\begin{theorem}\label{thm:xnlp-hardness-pathwidth}
\ShortestPathDiscovery{} is \XNLP-hard in the classical undirected
unweighted token-sliding model when parameterized by the pathwidth of the input
graph.
\end{theorem}

\shortlongversion{}{
\begin{proof}
We reduce from \textsc{Circu\-lating Orientation}. Let $(H,w)$ be an instance of
\textsc{Circu\-lating Orientation}, and suppose that $H$ is given with a path
decomposition of width at most $k$. Let
$e_1,\ldots,e_m$
be the edge order obtained from \Cref{lem:pathwidth-spine-closure}. For every
edge $e_i=u_iv_i$, we create an edge gadget $Q_i$ with source vertex $q_i^s$ and
target vertex $q_i^t$. The gadget consists of two internally vertex-disjoint
$q_i^s$-$q_i^t$ paths: $Q_i^{u_i}$ and $Q_i^{v_i}$.
Each of these two paths has exactly $w(e_i)$ internal vertices. We identify
$q_i^t$ with $q_{i+1}^s$ for all $i\in[m-1]$, and we add two terminals $s,t$
together with the edges
   $sq_1^s$ and $q_m^t t$.

Thus, every shortest $s$-$t$ path in the gadget chain traverses the gadgets
$Q_1,\ldots,Q_m$ in this order and chooses exactly one of the two branches in
each gadget. All such paths have the same number of vertices, namely $m+1+W$.

For every vertex $v\in V(H)$, we create a reservoir vertex $p_v$. For every
internal vertex $z$ of a branch associated with $v$, we connect $p_v$ to $z$ by
an internally vertex-disjoint path of length exactly
$L \coloneqq m+W+1$.
The value of $L$ is chosen larger than the length of the gadget chain, so no
shortest $s$-$t$ path uses a reservoir path. Therefore, the shortest $s$-$t$ paths
are precisely the paths that choose one branch in each edge gadget.

The initial token multiset $S$ is defined as follows. We place one token on every
mandatory chain vertex $q_1^s,q_1^t,q_2^t,\ldots,q_m^t,t$.
Moreover, for every vertex $v\in V(H)$, we place exactly $W(v)/2$ tokens on the
reservoir vertex $p_v$. Finally, we set $b \coloneqq W\cdot L$.

We prove correctness of the construction.

Every shortest $s$-$t$ path in the constructed graph must traverse the gadgets
$Q_1,\ldots,Q_m$ in order and must choose exactly one branch in each gadget. Since
each chosen branch contributes exactly $w(e_i)$ internal vertices, every shortest
$s$-$t$ path contains exactly $W$ internal branch vertices in addition to the
mandatory chain vertices.

The mandatory chain vertices already carry tokens initially. Hence, to discover
a shortest $s$-$t$ path, one has to occupy exactly the $W$ internal vertices of
the chosen branches. None of these internal branch vertices carries a token
initially. The only tokens not already placed on mandatory chain vertices are the
reservoir tokens.

Every reservoir token has distance exactly $L$ to each branch vertex associated
with its reservoir. Therefore, occupying the $W$ internal vertices of a chosen
shortest path costs at least $W\cdot L=b$. Since the budget is exactly $b$, every
solution of cost at most $b$ must move exactly $W$ reservoir tokens, each along a
shortest reservoir path of length $L$. In particular, no token can afford a detour
and no mandatory chain token can be moved away and replaced without exceeding the
budget.

Now consider an edge $e_i=u_iv_i$. If the final shortest path uses the branch
$Q_i^{u_i}$, then all~$w(e_i)$ internal vertices of this branch must be occupied
by tokens from the reservoir $p_{u_i}$. We interpret this as orienting $e_i$ from
$u_i$ to $v_i$. Similarly, if the final path uses the branch~$Q_i^{v_i}$, we
orient $e_i$ from $v_i$ to $u_i$.

For a vertex $v\in V(H)$, the number of tokens consumed from $p_v$ is exactly the
total weight of edges oriented out of $v$. Since $p_v$ initially contains exactly
$W(v)/2$ tokens and all reservoir tokens are used, we obtain
\[
   \sum_{e\text{ oriented out of }v} w(e)=W(v)/2.
\]
Thus, the orientation is circulating.

Conversely, suppose that $(H,w)$ admits a circulating orientation. For every edge
$e_i=u_iv_i$, choose the branch $Q_i^{u_i}$ if $e_i$ is oriented from $u_i$ to
$v_i$, and choose the branch $Q_i^{v_i}$ otherwise. This gives a shortest
$s$-$t$ path through the gadget chain. For every vertex $v$, the total number of
branch vertices associated with $v$ on this path is exactly the weighted
outdegree of $v$, which is $W(v)/2$. Hence, the tokens at $p_v$ suffice exactly.
Moving each such token along the corresponding length-$L$ reservoir path yields
total cost $W\cdot L=b$. Therefore, the constructed instance is positive if and
only if $(H,w)$ admits a circulating orientation.

It remains to bound the pathwidth. The edge gadgets are arranged along the spine
given by \Cref{lem:pathwidth-spine-closure}. For an edge $e_i=u_iv_i$, the gadget
$Q_i$ is placed at the position corresponding to the subdivision vertex of
$e_i$. The reservoir vertex $p_v$ is carried through the interval in which $v$
appears in the path decomposition of $H$. Since at every bag, there are at most
$k+1$ active vertices of $H$, at most $k+1$ reservoir vertices have to be kept
simultaneously. The two branches of the current edge gadget and the attachment
paths to the relevant reservoirs can be introduced and forgotten locally. This
adds only a constant number of vertices to each bag, apart from the active
reservoir vertices. Hence, the pathwidth of the constructed graph is bounded by a
function of $k$.

Thus, the reduction maps instances of \textsc{Circulating Orientation} of
pathwidth at most~$k$ to instances of \textsc{Shortest Path Discovery} whose
pathwidth is bounded by a function of $k$. Since the source problem is
\XNLP-hard parameterized by pathwidth, the theorem follows.
\end{proof}}

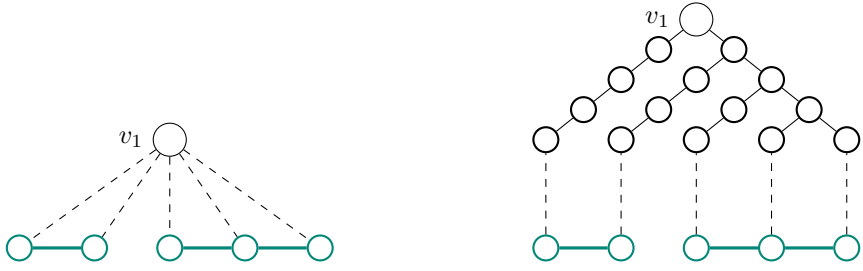
\begin{figure}
    \centering
    \begin{tikzpicture}[mainvertex/.style={draw, circle, minimum height=1.25em},hiddenmainvertex/.style={minimum height=1.3em}, weight/.style={},
    token/.style={draw,rectangle, minimum height=0.25em,  minimum width=0.25em, fill=black, inner sep=0em},
    longpath/.style={dashed},
    edge/.style={draw},
    pathedge/.style={draw,very thick,darkgreen,overlay},
    pathvertex/.style={draw=darkgreen, thick, circle, minimum width=0.25em},
    catervertex/.style={draw,thick, circle, minimum width=0.25em}
]
    \newcommand\yscale{0.4}
    \node[mainvertex,label=left:$v_1$] (v1) at (3,0*\yscale) {};

    \foreach \i in {1,...,5} {
        \node[pathvertex] (q\i) at (\i,-3.6*\yscale) {};
        \draw[longpath] (v1) to (q\i);
        \node[pathvertex] (c\i) at (\i+7,-3.6*\yscale) {};
        \node[catervertex] (d\i) at (\i+7,0) {};
        \draw[longpath] (c\i) to (d\i);
        \ifnum\i>1
        \ifnum\i=3
        \else
        \draw[pathedge] (q\i) to (q\the\numexpr\i-1\relax);
        \draw[pathedge] (c\i) to (c\the\numexpr\i-1\relax);
        \fi
        \fi
    }

    \node[catervertex] (e1) at (8.5,1*\yscale) {};
    \node[catervertex] (e2) at (9.5,1*\yscale) {};
    \node[catervertex] (e3) at (10.5,1*\yscale) {};
    \node[catervertex] (e4) at (11.5,1*\yscale) {};

    \draw (d1) to (e1);
    \draw (d2) to (e2);
    \draw (d3) to (e3);
    \draw (d4) to (e4);
    \draw (d5) to (e4);

    \node[catervertex] (f1) at (9,2*\yscale) {};
    \node[catervertex] (f2) at (10,2*\yscale) {};
    \node[catervertex] (f3) at (11,2*\yscale) {};

    \draw (e1) to (f1);
    \draw (e2) to (f2);
    \draw (e3) to (f3);
    \draw (e4) to (f3);

    \node[catervertex] (g1) at (9.5,3*\yscale) {};
    \node[catervertex] (g2) at (10.5,3*\yscale) {};
    \draw (f1) to (g1);
    \draw (f2) to (g2);
    \draw (f3) to (g2);

    \node[mainvertex,label=left:$v_1$] (w1) at (10,4*\yscale) {};

    \draw (g1) to (w1);
    \draw (g2) to (w1);

\end{tikzpicture}
    \caption{Replacing reservoir vertices with subdivided caterpillars}
    \label{fig:caterpillar}
\end{figure}

We now refine the construction to obtain bounded cutwidth. The only obstruction
to bounded cutwidth in the construction above is the possible large degree of the
reservoir vertices $p_v$. We replace each such high-degree reservoir by a
subdivided caterpillar as in \Cref{fig:caterpillar} while preserving the crucial property that every useful
move from the reservoir to an associated branch vertex has the same length exactly $L$.

\shortlongversion{}{
Fix a vertex $v\in V(H)$. Let
$Z_v$
be the set of all internal branch vertices associated with $v$; thus $|Z_v| = W(v)$.
We order the vertices of $Z_v$ according to the global edge order
$e_1,\ldots,e_m$, and arbitrarily inside one edge gadget. Write this order as
$z_1^v,\ldots,z_{W(v)}^v$.
Instead of connecting $p_v$ directly to every vertex of $Z_v$ by a length-$L$
path, we attach a subdivided caterpillar rooted at $p_v$. More precisely, we build a path
$p_v=c_0^v-c_1^v-\cdots-c_{W(v)}^v$.
Then, for each $j\in[W(v)]$, we connect $c_j^v$ to $z_j^v$ by a path of length
$L-j$. Hence, the distance from $p_v$ to $z_j^v$ is exactly $\dist(p_v,c_j^v)+(L-j)=j+(L-j)=L$.
Because $W(v)\le W<L$, all these lengths are positive. The resulting graph is a
subdivided caterpillar: the path
$p_v,c_1^v,\ldots,c_{W(v)}^v$
is the spine, and the paths to the branch vertices are subdivided legs. Every
vertex in this replacement has degree at most three.

The initial configuration is unchanged except that the $W(v)/2$ reservoir tokens
remain placed at the root $p_v$ of the corresponding caterpillar. The budget
remains $b=W\cdot L$.
The correctness proof is unchanged: every branch vertex associated with $v$ is
still at distance exactly $L$ from $p_v$, and any detour costs more than the
available tight budget.}

\begin{theorem}\label{thm:xnlp-hardness-cutwidth}
\BothPathDiscovery{} is \XNLP-hard in the classical undirected
unweighted token-sliding model when parameterized by the cutwidth of the input
graph.
\end{theorem}

\shortlongversion{}{
\begin{proof}
We use the construction from \Cref{thm:xnlp-hardness-pathwidth}, but replace each
reservoir vertex by the subdivided caterpillar described above.

Correctness is exactly as before. The replacement preserves the property that,
for every vertex $v\in V(H)$ and every internal branch vertex associated with
$v$, the distance from the reservoir root $p_v$ to that branch vertex is exactly
$L$. Therefore, every feasible solution of cost at most $W\cdot L$ must move
exactly $W$ reservoir tokens, each over distance exactly $L$, and the chosen
branches again encode a circulating orientation of $(H,w)$.

It remains to prove that the constructed graph has cutwidth bounded by a function
of the pathwidth $k$ of the source graph. We describe a linear layout. Start with
the path decomposition of $H$ used in \Cref{lem:pathwidth-spine-closure}, and use
the induced edge order $e_1,\ldots,e_m$. We place the edge gadgets in this order.
For every vertex $v\in V(H)$, the vertices
$c_0^v,c_1^v,\ldots,c_{W(v)}^v$
of the reservoir spine are distributed along the interval of the path
decomposition in which $v$ is active, in the same order as the corresponding
branch vertices $z_1^v,\ldots,z_{W(v)}^v$ occur along the gadget chain. The
subdivided hair from $c_j^v$ to $z_j^v$ is placed locally next to the branch
vertex~$z_j^v$.

Consider any cut in this layout. The only long edges or paths that may cross the
cut are the reservoir spines of vertices $v$ whose activity interval contains the
current position. Since the source path decomposition has width $k$, there are
at most $k+1$ such vertices. Each active reservoir spine contributes at most one
spine edge crossing the cut. Each local hair contributes only constantly many
crossing edges, because its internal vertices are placed consecutively next to
the branch vertex to which it attaches. The edge gadget currently intersected by
the cut contributes only a constant number of edges. Hence, the number of edges
crossing any cut is bounded by a function of $k$.

Equivalently, the construction has pathwidth bounded by a function of $k$ and
maximum degree bounded by a constant, so by the inequalities above its cutwidth
is bounded by a function of $k$. Therefore, the reduction is a parameterized
reduction to \textsc{Shortest Path Discovery} parameterized by cutwidth. Since
\textsc{Circulating Orientation} is \XNLP-hard parameterized by pathwidth, the
claim follows.
\end{proof}}

\subsection{Further hardness results in the weighted two-graph model}
\label{subsec:two-graph-hardness}

We now record further hardness consequences of the same construction with slight variations.
Throughout this subsection, the problem graph $G$ and the movement graph $M$ may differ. We
use the construction from the proof of
\Cref{thm:xnlp-hardness-pathwidth}.
The reductions below only modify the weights, structure and/or the movement graph. The
correctness argument is always the same. 

The same reductions can also be made acyclic in the directed two-graph model.

\begin{proposition}\label[proposition]{prop:dfvs}
   \PathDiscovery in the directed two-graph model is para-NP-hard when
   parameter\-ized by the directed feedback vertex set of the union of problem graph $G$ and movement graph $M$.
\end{proposition}

\shortlongversion{}{
   \begin{proof}[Proof sketch]
Orient every edge of the problem graph from left to right along the gadget chain,
that is, from~$s$ towards $t$. Orient every movement edge, or every subdivided
movement path, from the reservoir side towards the corresponding gadget vertex.
Then both $G$ and $M$ are DAGs. In particular, their directed feedback vertex set
numbers are $0$.
\end{proof}}

\begin{proposition}\label{prop:path-diameter-positive-hard-new}
\PathDiscovery in the two-graph model is para-NP-hard when
parameter\-ized by the weighted diameter of the problem graph $G$, even if all edge
weights of $G$ are positive integers. In fact, the hardness already holds for
instances in which the weighted diameter of~$G$ is bounded by an absolute
constant.
\end{proposition}

\shortlongversion{}{
\begin{proof}[Proof sketch]
We start from the construction of \Cref{thm:xnlp-hardness-pathwidth}. The movement graph is kept as
in the basic reduction, so the reservoir accounting remains unchanged. We modify
only the problem graph in order to make its weighted diameter constant. This is
done by adding a diameter-control structure with positive edge weights which
places all vertices at constant weighted distance from each other, while not
creating any new reachable $s$-$t$.

Since the movement graph and the token placement are unchanged, any solution
still has to occupy the internal vertices of one branch in each edge gadget.
Thus, a solution again determines, and is determined by, a circulating orientation
of the source instance. The added diameter-control edges have only the purpose of
bounding the weighted diameter of $G$ by an absolute constant. Hence, the problem
is NP-hard already for constant weighted diameter, and therefore para-NP-hard for
this parameter.
\end{proof}}

\begin{proposition}\label{prop:diam-zero-hard-new}
If zero weights are allowed in the problem graph $G$, then \BothPathDiscovery is para-NP-hard
when parameterized by the weighted diameter of $G$. In fact, this already holds
for instances with weighted diameter $0$.
\end{proposition}

\shortlongversion{}{
\begin{proof}[Proof sketch]
We use the same construction, and assign weight $0$ to every edge of the
problem graph $G$. Then the weighted diameter of $G$ is $0$. The movement graph,
the initial token placement, and the budget are unchanged.

For \textsc{Path Discovery}, edge weights in $G$ do not affect feasibility.
For \textsc{Shortest Path Discovery}, every $s$-$t$ path has weight $0$ and is
therefore shortest. Thus, the set of relevant branch-choice paths is unchanged.
The standard reservoir accounting then gives a one-to-one correspondence between
solutions and circulating orientations. Hence, both problems are NP-hard already
when the weighted diameter of $G$ is $0$.
\end{proof}}

\begin{theorem}\label{thm:b-zero-hard-new}
If zero weights are allowed in the movement graph $M$, then \BothPathDiscovery is para-NP-hard
when parameterized by the budget $b$. In fact, this already holds for $b=0$.
\end{theorem}

\shortlongversion{}{
\begin{proof}[Proof sketch]
Again use the same problem graph $G$ and the same initial token placement. For
every reservoir vertex $p_v$ and every internal branch vertex associated with
$v$, add a movement edge, or a subdivided movement path, from $p_v$ to that
branch vertex of total weight $0$. No other movement edges are needed. Set
$b:=0$.

Thus, a token can be moved for free from the reservoir of $v$ to any branch vertex
associated with $v$, but not to branch vertices associated with other vertices.
Consequently, the budget no longer measures distance, but the movement graph still
enforces the endpoint choice. A discovered path therefore chooses one branch in
each edge gadget, and the chosen branches can be filled with budget $0$ exactly
when every reservoir is used exactly up to its capacity. This is precisely the
circulating-orientation condition. The same argument applies to
\textsc{Shortest Path Discovery}, since the two branches of each gadget have
equal length and hence all branch-choice paths are shortest $s$-$t$ paths.
Therefore, the problems are NP-hard already for $b=0$.
\end{proof}}

\section{Conclusion}

We introduced a directed weighted two-graph model for solution discovery, in which the
graph defining feasibility is separated from the graph defining movement.  This separation
captures situations in which the combinatorial structure of the desired solution and the
constraints governing how tokens, agents, or resources may move are inherently different.
Using \textsc{Path Discovery} and \textsc{Shortest Path Discovery} as test cases, we obtained
a varied complexity landscape: the model admits efficient algorithms in several natural
settings, including for bounded numbers of tokens, token jumping, bounded solution size,
and small feedback edge set, while it remains hard under several structural restrictions and
for important parameters.

These results indicate that the two-graph perspective is not merely a technical
generalization of the classical token-sliding and token-jumping models, but a useful
framework for under\-standing how feasibility and movement interact.  A natural direction
for future work is to study further discovery problems in this model, such as variants of
vertex cover, domination, matching, cut, or clustering problems, and to determine which
features of the feasibility problem and of the movement graph govern tractability.  It would
also be interesting to investigate the same separation of feasibility and movement in the
classical reconfiguration setting, where both the initial and target solutions are prescribed,
and all intermediate configurations are required to remain feasible.  We expect that this
two-graph viewpoint will lead to new algorithmic questions and to a more refined
understanding of reconfiguration and discovery under realistic movement constraints.


\bibliographystyle{plainurl}
\bibliography{ref}
\end{document}